%% file: wl-rounds.tex
\documentclass[a4paper,11pt]{article}
\usepackage[margin = 1in]{geometry}
\usepackage[utf8]{inputenc}
\usepackage[T1]{fontenc}
\usepackage[dvipsnames]{xcolor}
\usepackage{graphicx}
\usepackage{amsmath, amssymb, amstext, mathtools}
\usepackage{amsthm}
\usepackage{tikz}
\usepackage{pgfplots}
\usepackage{float}
\usepackage{hyperref}
\usepackage{paralist}
\usepackage{enumitem}
\usepackage[titlenumbered, linesnumbered, ruled]{algorithm2e}

\allowdisplaybreaks

\definecolor[named]{urlblue}{cmyk}{1,0.58,0,0.21}

\hypersetup{
 breaklinks=true,
 colorlinks=true,
 citecolor=purple!70!blue!60!black,
 linkcolor=purple!70!blue!90!black,
 urlcolor=urlblue,
 pdflang={en},
 pdftitle={The Iteration Number of the Weisfeiler-Leman Algorithm},
 pdfauthor={Martin Grohe, Moritz Lichter and Daniel Neuen}
}

\usetikzlibrary{patterns,arrows,decorations.pathreplacing,decorations.pathmorphing,calc}

\tikzstyle{smallvertex}=[draw,circle,fill=white,minimum size=4pt,inner sep=0pt]

\newtheorem{theorem}{Theorem}[section]
\newtheorem{lemma}[theorem]{Lemma}
\newtheorem{corollary}[theorem]{Corollary}

\newtheorem{observation}[theorem]{Observation}

\theoremstyle{definition}
\newtheorem{definition}[theorem]{Definition}

\theoremstyle{remark}
\newtheorem{remark}[theorem]{Remark}
\newtheorem{claim}[theorem]{Claim}

\newenvironment{claimproof}{\begin{proof}}{\end{proof}}

\newcommand{\WL}[2]{\chi^{(\infty)}_{#1}[#2]}
\newcommand{\WLit}[3]{\chi^{(#2)}_{#1}[#3]}
\newcommand{\refWL}[2]{\operatorname{step}_{#1}\!\left(#2\right)}

\DeclareMathOperator{\typ}{{\sf typ}}
\DeclareMathOperator{\atp}{atp}

\newcommand{\NN}{\mathbb N}
\newcommand{\FF}{\mathbb F}

\newcommand{\CC}{\mathcal C}
\newcommand{\CD}{\mathcal D}
\newcommand{\CF}{\mathcal F}
\newcommand{\CG}{\mathcal G}

\newcommand{\CM}{\mathcal M}
\newcommand{\CP}{\mathcal P}
\newcommand{\CQ}{\mathcal Q}
\newcommand{\CV}{\mathcal V}

\newcommand{\FA}{\mathfrak A}
\newcommand{\FB}{\mathfrak B}
\newcommand{\FC}{\mathfrak C}

\newcommand{\logic}[1]{{\sf #1}}
\newcommand{\FO}{\logic{FO}}
\newcommand{\LC}{\logic C}
\newcommand{\LL}{\logic L}
\newcommand{\LCk}[1]{\LC_{#1}}

\newcommand{\LCkq}[2]{\LC_{#1}^{(#2)}}
\newcommand{\LLk}[1]{\LL_{#1}}

\newcommand{\LLkq}[2]{\LL_{#1}^{(#2)}}

\newcommand{\hi}{{\sf hi}}
\newcommand{\lo}{{\sf lo}}

\newcommand{\Complex}{{\mathbb C}}
\newcommand{\Alg}{\mathbb A}
\newcommand{\FullMatrixAlg}{\mathsf M}

\renewcommand{\vec}[1]{\boldsymbol{#1}}

\newcommand{\Binary}{\operatorname{Bin}}

\DeclareMathOperator{\qr}{qr}

\DeclareMathOperator{\cl}{cl}
\DeclareMathOperator{\attr}{attr}

\DeclareMathOperator{\spn}{span}

\newcommand{\ceil}[1]{\left\lceil#1\right\rceil}

\newcommand{\angles}[1]{\left\langle#1\right\rangle}

\newcommand{\bigmid}{\mathrel{\big|}}
\newcommand{\Bigmid}{\mathrel{\Big|}}
\newcommand{\Biggmid}{\mathrel{\Bigg|}}

\usepackage[textwidth=2.5cm]{todonotes}

\newcommand{\orcid}[1]{\href{https://orcid.org/#1}{\includegraphics[height=1.8ex]{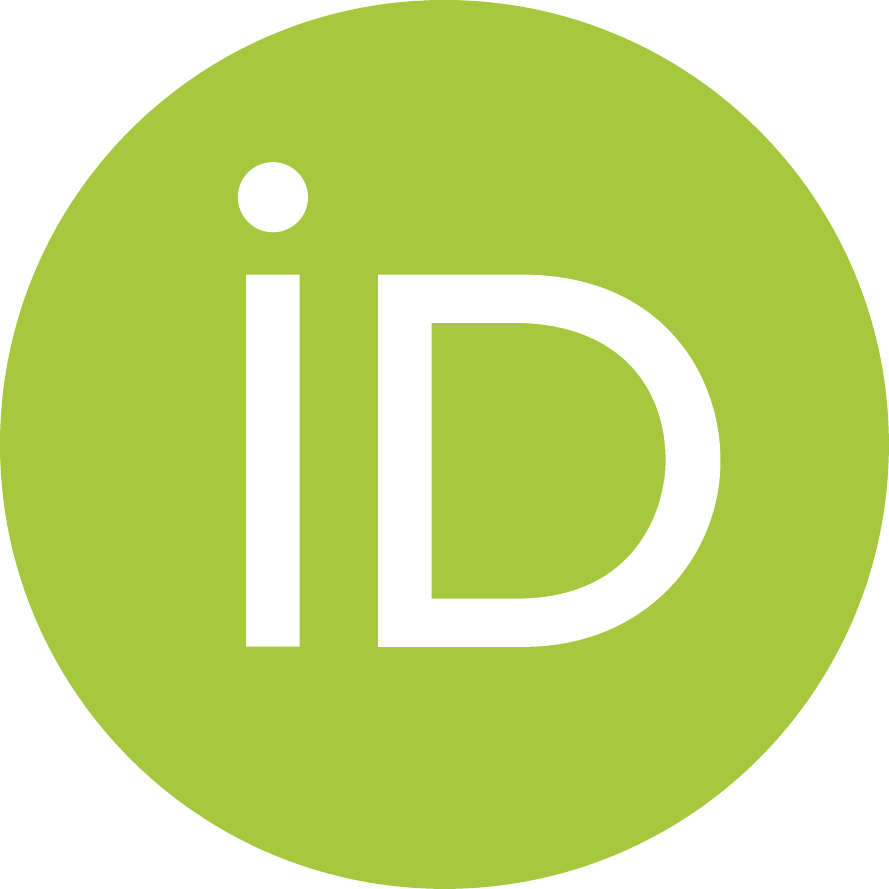}}}
\newcommand{\email}[1]{\href{mailto:#1}{\texttt{#1}}}

\title{The Iteration Number of the Weisfeiler-Leman Algorithm}
\author{
Martin Grohe \orcid{0000-0002-0292-9142}\\
RWTH Aachen University\\
\email{grohe@informatik.rwth-aachen.de}
\and
Moritz Lichter \orcid{0000-0001-5437-8074}\\
TU Darmstadt\\
\email{lichter@mathematik.tu-darmstadt.de}
\and
Daniel Neuen \orcid{0000-0002-4940-0318}\\
Simon Fraser University\\
\email{dneuen@sfu.ca}
}
\date{}

\begin{document}

\maketitle

\begin{abstract}
 We prove new upper and lower bounds on the number of iterations the $k$-dimensional Weisfeiler-Leman algorithm ($k$-WL) requires until stabilization.
 For $k \geq 3$, we show that $k$-WL stabilizes after at most $O(kn^{k-1}\log n)$ iterations (where $n$ denotes the number of vertices of the input structures), obtaining the first improvement over the trivial upper bound of $n^{k}-1$ and extending a previous upper bound of $O(n \log n)$ for $k=2$ [Lichter et al., LICS 2019].
 
 We complement our upper bounds by constructing $k$-ary relational
 structures on which $k$-WL requires at least $n^{\Omega(k)}$
 iterations to stabilize.  This improves over a previous lower bound
 of $n^{\Omega(k / \log k)}$ [Berkholz, Nordstr{\"{o}}m, LICS 2016].
 
 We also investigate tradeoffs between the dimension and the iteration number of WL, and show that $d$-WL, where $d = \lceil\frac{3(k+1)}{2}\rceil$, can simulate the $k$-WL algorithm using only $O(k^2 \cdot n^{\lfloor k/2\rfloor + 1} \log n)$ many iterations, but still requires at least $n^{\Omega(k)}$ iterations for any $d$ (that is sufficiently smaller than $n$).
 
 The number of iterations required by $k$-WL to distinguish two structures corresponds to the quantifier rank of a sentence distinguishing them in the $(k + 1)$-variable fragment $\LCk{k+1}$ of first-order logic with counting quantifiers.
 Hence, our results also imply new upper and lower bounds on the quantifier rank required in the logic $\LCk{k+1}$, as well as tradeoffs between variable number and quantifier rank.
\end{abstract}

\section{Introduction}
\label{sec:introduction}
\input{introduction}

\section{Preliminaries}
\label{sec:preliminaries}
\input{preliminaries}

\section{Upper Bounds}
\label{sec:upper-bounds}
\input{upper-bounds}

\section{Long Sequences of Stable Colorings}
\label{sec:long-sequences}
\input{long-sequences}

\section{Lower Bounds on the Iteration Number of WL}
\label{sec:lower-bounds}
\input{lower-bounds}

\section{Trading Variable Number for Quantifier Depth}
\label{sec:tradeoffs}
\input{tradeoffs}

\section{Conclusion}
\label{sec:conclusion}
\input{conclusion}

\bibliographystyle{plainurl}
\small
\bibliography{literature}

\end{document}

%% file: introduction.tex
The Weisfeiler-Leman (WL) algorithm is a combinatorial algorithm that, given a relational structure $\FA$ (in most applications, this structure is a graph), iteratively computes an isomorphism-invariant coloring of tuples of vertices of $\FA$.
The original algorithm introduced by Weisfeiler and Leman \cite{WeisfeilerL68} is the $2$-dimensional version that colors pairs of vertices.
Its generalization to arbitrary dimension $k \geq 1$, independently introduced by Babai and Mathon as well as Immerman and Lander \cite{ImmermanL90} (see also \cite{Babai16} for a historic note), 
yields for every natural number $k$ the $k$-dimensional WL algorithm ($k$-WL), which iteratively refines a coloring of vertex $k$-tuples by aggregating local structural information encoded in the colors.
More concretely, the $k$-WL algorithm initially colors all $k$-tuples of vertices $\vec v = (v_1,\dots,v_k)$ of a structure $\FA$ by the isomorphism type of the underlying induced ordered substructure.
Afterwards, in each iteration, the coloring is refined by taking the colors of all tuples into account that can be obtained from $\vec v$ by replacing a single entry of the tuple.
This process necessarily stabilizes after a finite number of iterations and the resulting coloring can be used to classify $k$-tuples of vertices.

The most prominent application of the WL algorithm lies in the context of the graph isomorphism problem.
Indeed, since no isomorphism between two structures $\FA$ and $\FB$ can map tuples of vertices of different colors to each other, the WL algorithm provides a hierarchy of increasingly powerful heuristics to the graph isomorphism problem.
While there is no dimension $k$ for which $k$-WL serves as a complete isomorphism test \cite{CaiFI92}, the algorithm is still surprisingly powerful.
For example, Grohe \cite{Grohe17} proved that for every non-trivial minor-closed graph class there is some $k \in \NN$ such that $k$-WL computes a different coloring on all non-isomorphic graphs, and thus provides a polynomial-time isomorphism test on that class.
Moreover, the WL algorithm is also regularly used as a subroutine in isomorphism algorithms (see, e.g., \cite{Neuen21,Neuen22,SunW15}) which includes Babai's \cite{Babai16} quasipolynomial-time graph isomorphism test that employs the WL algorithm with dimension $k = O(\log n)$.

More recently, the WL algorithm has also received significant attention in the machine learning context where it characterizes the expressiveness of graph neural networks \cite{Grohe21,MorrisRFHLRG19,XuHLJ19} and, more generally, the colorings computed by WL are used in classification tasks on graph-structured data sets (see, e.g., \cite{MorrisLMRKGFB21,ShervashidzeSLMB11}).

Since the late 1980s, the WL algorithm has played an important role in descriptive complexity theory.
Indeed, it was independently introduced in the context of descriptive complexity by Immerman and Lander~\cite{ImmermanL90}.
The main reason for this is that $k$-WL can be seen as an equivalence test for the logic $\LCk{k+1}$, the $(k+1)$-variable fragment of first-order logic with counting quantifiers $\exists^{\ge n}x$.
Through this connection, the algorithm has turned out to be important for studying the expressiveness of fixed-point logic with counting \cite{CaiFI92} and, more generally, for the quest for a logic capturing polynomial time \cite{Grohe08,Otto17}.

In this work, we study the iteration number of $k$-WL, i.e., the number of iterations the algorithm requires until stabilization.
Since the number of color classes increases in each iteration, the $k$-WL algorithm trivially requires at most $n^{k}-1$ rounds to stabilize.
For $k = 1$, Kiefer and McKay \cite{KieferM20} proved that this trivial bound is optimal by providing several infinite families of graphs $G$ for which $1$-WL requires $n-1$ iterations to stabilize (where $n$ denotes the number of vertices of $G$).
In contrast, for $k = 2$, Lichter, Ponomarenko and Schweitzer \cite{LichterPS19} (improving an earlier upper bound by Kiefer and Schweitzer~\cite{KieferS19}) obtained an upper bound of $O(n \log n)$ on the iteration number of $2$-WL.
Beyond that, no improved upper bounds are known for $k \geq 3$.
As our first main contribution, we obtain non-trivial bounds on the iteration number of $k$-WL for all $k \geq 2$.

\begin{theorem}
 \label{thm:wl-round-upper-bound}
 For all $k\ge 2$, the $k$-dimensional Weisfeiler-Leman algorithm stabilizes after $O(kn^{k-1}\log n)$ refinement rounds on all relational structures $\FA$ of arity at most $k$ where $n$ denotes the size of the universe.
\end{theorem}

For the proof, we extend the algebraic arguments from \cite{LichterPS19}.
Consider a structure $\FA$ with vertex set $V$ of size $n$ and let $\chi_0,\dots,\chi_\ell \colon V^{k} \rightarrow C$ denote the sequence of colorings computed by $k$-WL, i.e., $\chi_i$ is the coloring computed in the $i$-th iteration.
For $k = 2$, Lichter et al.\ \cite{LichterPS19} associate with each coloring $\chi_i$ a matrix algebra as follows.
For each color $c$ in the image of $\chi_i$, let $M_{i,c}$ denote the $V \times V$ indicator matrix that sets $M_{i,c}(v_1,v_2) \coloneqq 1$ if $\chi_i(v_1,v_2) = c$, and $M_{i,c}(v_1,v_2) \coloneqq 0$ otherwise.
The matrices $M_{i,c}$, where $c$ ranges over all colors in the image of $\chi_i$, generate a matrix algebra $\Alg^{(i)}$ of $V \times V$ matrices over the complex numbers using standard matrix multiplication.
Using representation-theoretic arguments, it is possible to bound the length of the sequence of matrix algebras generated this way which eventually leads to the upper bound of $O(n \log n)$.

The proof of Theorem \ref{thm:wl-round-upper-bound} follows a similar strategy.
For each color in the image of $\chi_i$, we obtain an indicator tensor $M_{i,c} \in \Complex^{V^{k}}$.
Now, the key challenge in generalizing the arguments of \cite{LichterPS19} is to define a suitable multiplication of those tensors that can be ``simulated'' by a single round of $k$-WL.
Given such a multiplication, we then show that the generated algebra $\Alg^{(i)}$ is isomorphic to a subalgebra of the $n^{k-1} \times n^{k-1}$ full matrix algebra (over the complex numbers) which then again allows us to use algebraic arguments to obtain the desired upper bound.

Our arguments actually prove a more general result.
Let $\chi_0,\dots,\chi_\ell \colon V^{k} \rightarrow C$ be a sequence of finer and finer colorings (i.e., the partition into color classes of $\chi_{i}$ refines the partition into color classes of $\chi_{i-1}$ for all $i \in [\ell]$) where in each step the coloring is refined at least as much as by a single iteration of $k$-WL.
Then the length of the sequence is bounded by $\ell = O(kn^{k-1}\log n)$.
As a lower bound to our arguments, we show that, in this more general setting, our upper bound is tight up to a factor $O_k(\log n)$ (the $O_k(\cdot)$-notation hides constant factors in $k$).
Here, the key insight is that we can find a sequence of finer and finer colorings $\chi_0,\dots,\chi_\ell \colon V^{k} \rightarrow C$ of length $\Omega_k(n^{k-1})$ that are all stable with respect to $k$-WL.
As such, it provides a lower bound in the more general setting explained above (but it does not give any lower bounds on the iteration number of $k$-WL) and implies that new ideas are likely required to obtain further improvements on the upper bounds of the iteration number of $k$-WL (see Section \ref{sec:long-sequences} for more details).

Looking for lower bounds on the iteration number of $k$-WL, F{\"{u}}rer \cite{Furer01} provided, for every $k \geq 2$, a family of graphs on which $k$-WL requires at least $\Omega(n)$ many iterations until stabilization.
For $k$ sufficiently large, this result was strengthened by Berkholz and Nordstr{\"{o}}m \cite{BerkholzN16} who constructed $k$-ary relational structures $\FA$ of size $n$ on which $k$-WL requires at least $n^{\Omega(k/ \log k)}$ many iterations.
Answering an open question from \cite{BerkholzN16}, our second main contribution is an improved lower bound that gets rid of the $1/\log k$ factor in the exponent.
Actually, we prove the following even stronger result.

\begin{theorem}
 \label{thm:wl-round-lower-bound}
 There are absolute constants $k_0 \in \NN$ and $\alpha,\varepsilon > 0$ such that for every $d \geq k \geq k_0$ and every $n \geq \alpha d^8 k^6$
 there is a is pair of $k$-ary relational structures $\FA$ and $\FB$ of size $|V(\FA)| = |V(\FB)| = n$ that are distinguished by $k$-WL, but $d$-WL does not distinguish $\FA$ and $\FB$ after $n^{\varepsilon k}$ refinement rounds.
\end{theorem}

We note that, as in the work of Berkholz and Nordstr{\"{o}}m \cite{BerkholzN16}, the structures we need to prove this theorem are \emph{$k$-ary}, that is, have relations of arity $k$.

The structures $\FA$ and $\FB$ provided by the theorem can be distinguished by $k$-WL which trivially requires at most $n^{k}-1$ rounds.
The theorem states that, even if we are allowed to increase the dimension of the Weisfeiler-Leman algorithm to $d$, the structures can still not be distinguished unless $d$-WL runs for at least $n^{\varepsilon k}$ rounds.
This result stands in strong contrast to several existing results for restricted classes of graphs.
For example, $k$-WL distinguishes between all non-isomorphic pairs of graphs of tree-width at most $k$ \cite{KieferN22}, and increasing the dimension to $4k+3$ guarantees that $O(\log n)$ iterations suffices to distinguish between all non-isomorphic pairs of graphs of tree-width at most $k$ \cite{GroheV06}.
Similar results are known for planar graphs \cite{GroheK21,Verbitsky07}.
The above theorem rules out such results for general relational structures even if we only wish to improve the iteration number to, for example, linear in $n$.

By setting $d = k$, we obtain the following corollary which shows that the upper bound in Theorem \ref{thm:wl-round-upper-bound} is optimal up to a constant factor (that does not depend on $k$) in the exponent.

\begin{corollary}
 There are absolute constants $k_0 \in \NN$ and $\alpha,\varepsilon > 0$ such that for every $k \geq k_0$ and every $n \geq \alpha k^{14}$
 there is a $k$-ary structure $\FA$ of size $|V(\FA)| = n$ such that the $k$-dimensional Weisfeiler-Leman algorithm does not stabilize within $n^{\varepsilon k}$ refinement rounds on $\FA$.
\end{corollary}

For the proof of Theorem \ref{thm:wl-round-lower-bound}, our main technical contribution is to show that there is a $k_0 \in \NN$ such that, for all $d \geq k_0$, there are structures $\FA$ and $\FB$ of size $n$ that are distinguished by $k_0$-WL, but $d$-WL still requires $\Omega(n/d^{2})$ many iterations to distinguish $\FA$ and $\FB$.
Afterwards, we obtain Theorem \ref{thm:wl-round-lower-bound} by using a known hardness condensation \cite{BerkholzN16} that reduces the size of the structures while roughly preserving the number of iterations required to distinguish them.

Let us point out that F{\"{u}}rer \cite{Furer01} constructed graphs $G$ and $H$ which are distinguished by $k_0$-WL after $\Omega(n)$ many rounds.
However, as F{\"{u}}rer also shows, his instances are distinguished by $(3k_0)$-WL after only $O(\log n)$ many rounds which means that we cannot use them for our purposes.
Berkholz and Nordstr{\"{o}}m \cite{BerkholzN16} provided, for all $d \geq 2$, structures $\FA$ and $\FB$ of size $n$ that are distinguished by $2$-WL, but $d$-WL still requires $\Omega(n^{1/(1+\log d)})$ many rounds to distinguish them.
In combination with the hardness condensation, this leads to the previous lower bound of $n^{\Omega(k/ \log k)}$.

For the construction of our structures, we introduce the notion of \emph{layered expanders} whose global structure is similar to a $(k \times n)$-grid, but that locally (when looking at $O(k)$ consecutive columns) behave like an expander graph.
We then obtain propositional XOR-formulas from layered expanders which can be transformed into relational structures which satisfy the desired properties.

\paragraph{Connection to Logics.}

As pointed out above, $k$-WL is an equivalence test for the logic $\LCk{k+1}$.
That is, $k$-WL distinguishes between two structures $\FA$ and $\FB$ if and only if there is a sentence $\varphi \in \LCk{k+1}$ such that $\FA \models \varphi$ and $\FB \not\models \varphi$.
Additionally, the minimal quantifier rank of such a sentence equals (up to an additive error of at most $k$) the number of iterations $k$-WL requires to distinguish between $\FA$ and $\FB$.
With this in mind, Theorem \ref{thm:wl-round-upper-bound} can be reformulated as follows.

\begin{corollary}
 Let $k \geq 3$.
 Let $\FA$ and $\FB$ be two relational structures of arity at most $k$ that can be distinguished by a sentence in $\LCk{k}$.
 Then there is a sentence $\varphi \in \LCk{k}$ of quantifier rank at most $q = O(kn^{k-2}\log n)$ such that $\FA \models \varphi$ and $\FB \not\models \varphi$.
\end{corollary}

Similarly, we can reformulate Theorem \ref{thm:wl-round-lower-bound}, but here it turns out that we can obtain an even stronger result since the structures constructed in the theorem can already be distinguished in the logic $\LLk{k+1}$, the $(k+1)$-variable fragment of first-order logic \emph{without} counting quantifiers.

\begin{theorem}
 \label{thm:quantifier-rank-lower-bound}
 There are absolute constants $k_0 \in \NN$ and $\alpha,\varepsilon > 0$ such that for every $d \geq k \geq k_0$ and every $n \geq \alpha d^8 k^6$
 there is a pair of $k$-ary structures $\FA$ and $\FB$ of size $|V(\FA)| = |V(\FB)| = n$ that can be distinguished by a sentence in $k$-variable first-order logic $\LLk{k}$, but satisfy the same sentences in $\LLk{d}$ and $\LCk{d}$ up to quantifier rank $n^{\varepsilon k}$.
\end{theorem}

Hence, we obtain lower bounds for the quantifier rank not only for the logic $\LCk{k}$, but also for the logic $\LLk{k}$.
We stress that the lower bounds on the quantifier rank remain valid even if we arbitrarily increase the number of variables to any number $d$ (as long as $d$ is sufficiently far away from the size of the structures).
In other words, even if we are allowed to increase the number of variables, we cannot in general hope for significant improvements on the quantifier rank required to distinguish between two structures.

Having said that, our final result shows that at least some improvements on the upper bound are possible if we are allowed to increase the number of variables by roughly a factor of $3/2$.

\begin{theorem}
 \label{thm:trading-upper-bound-intro}
 Let $k \geq 2$.
 Let $\FA$ and $\FB$ be two relational structures of arity at most $k$ such that $n \coloneqq |V(\FA)| = |V(\FB)|$.
 Also suppose there is a sentence $\varphi \in \LCk{k+1}$ such that $\FA \models \varphi$ and $\FB \not\models \varphi$.
 Let $d \coloneqq \lceil\frac{3(k+1)}{2}\rceil$.
 Then there is a sentence $\psi \in \LCkq{d}{q}$ of quantifier rank $q = O(k^2 \cdot n^{\lfloor k/2\rfloor + 1} \log n)$ such that $\FA \models \psi$ and $\FB \not\models \psi$.
\end{theorem}

\paragraph{Structure of the Paper.}

After introducing the necessary preliminaries in the next section, we prove Theorem \ref{thm:wl-round-upper-bound} in Section \ref{sec:upper-bounds}.
Afterwards, we prove limitations of our approach to obtain improved upper bounds on the iteration number in Section \ref{sec:long-sequences}.
In Section \ref{sec:lower-bounds}, we obtain the lower bounds on the iteration number of WL and prove Theorems \ref{thm:wl-round-lower-bound} and \ref{thm:quantifier-rank-lower-bound}.
Finally, Theorem \ref{thm:trading-upper-bound-intro} is proved in Section \ref{sec:tradeoffs}.

%%% Local Variables:
%%% mode: latex
%%% TeX-master: "wl-rounds"
%%% End:

%% file: preliminaries.tex
We use $\NN = \{1,2,3,\dots\}$ to denote the positive integers.
For $n \in \NN$ we write $[n] \coloneqq \{1,\dots,n\}$ and $[0,n] \coloneqq \{0,\dots,n\}$.

\paragraph{Graphs.}

We use standard graph notation.
A \emph{graph} is a pair $G = (V(G),E(G))$ with finite vertex set $V(G)$ and edge set $E(G)$.
In this paper, all graphs are simple (i.e., there are no loops or multiedges) and undirected.
We write $vw$ to denote an edge $\{v,w\} \in E(G)$.
The \emph{(open) neighborhood} of a vertex $v \in V(G)$ is the set $N_G(v) \coloneqq \{w \in V(G) \mid vw \in E(G)\}$.
The degree of a vertex, denoted by $\deg_G(v)$, is the size of its neighborhood.
For $X \subseteq V(G)$ we define $N_G(X) \coloneqq (\bigcup_{v \in X}N_G(v)) \setminus X$ to denote the neighborhood of $X$.
If the graph $G$ is clear from context, we usually omit the index $G$ and simply write $N(v)$, $\deg(v)$ and $N(X)$.
For $X \subseteq V(G)$ we also write $G[X]$ to denote the subgraph of $G$ induced by $X$.

\paragraph{Relational Structures.}

In this work, we restrict ourselves to relational vocabularies (signatures) $\sigma = \{R_1,\dots,R_m\}$ where each $R_i$ is a relation symbol of a prescribed arity $k_i \geq 1$.
We say that $\sigma$ \emph{has arity at most $k$} if $k_i \leq k$ for all $i \in [m]$.
A \emph{$\sigma$-structure} is a tuple $\FA = (V(\FA),R_1^{\FA},\dots,R_m^{\FA})$ where $V(\FA)$ is a finite \emph{universe} and $R_i^{\FA} \subseteq (V(\FA))^{k_i}$ is a relation of arity $k_i$.
In the remainder of this work, we usually do not explicitly refer to the vocabulary underlying a structure $\FA$.
With this in mind, we say a structure $\FA = (V(\FA),R_1^{\FA},\dots,R_m^{\FA})$ \emph{has arity at most $k$} if the underlying vocabulary has arity at most $k$.

For $X \subseteq V(\FA)$ we define $\FA[X]$ to be the \emph{induced substructure of $\FA$ on $X$}, i.e., $\FA[X]$ is the relational structure with $V(\FA[X]) = X$ and
\[R_i^{\FA[X]} = R_i^{\FA} \cap X^{k_i}\]
for all $i \in [m]$.
Let $\FB = (V(\FB),R_1^{\FB},\dots,R_m^{\FB})$ be a second structure (over the same vocabulary $\sigma$).
An \emph{isomorphism} from $\FA$ to $\FB$ is a bijection $f\colon V(\FA) \rightarrow V(\FB)$ such that, for all $i \in [m]$ and all $v_1,\dots,v_{k_i} \in V(\FA)$, it holds that
\[(v_1,\dots,v_{k_i}) \in R_i^{\FA} \iff (f(v_1),\dots,f(v_{k_i})) \in R_i^{\FB}.\]
The structures $\FA$ and $\FB$ are \emph{isomorphic} if there is an isomorphism from $\FA$ to $\FB$.

\paragraph{Logics.}

Next, we cover bounded-variable fragments of first-order logic (with counting quantifiers).
Let $\sigma = \{R_1,\dots,R_m\}$ be a relational vocabulary and suppose $R_i$ has arity $k_i \geq 1$.
We write $\FO$ to denote standard \emph{first-order logic} defined via the grammar
\[\varphi ::= x_1 = x_2 ~|~ R_i(x_1,\dots,x_{k_i}) ~|~ \varphi \wedge \varphi ~|~ \neg \varphi ~|~ \exists x_1 \varphi\]
for all $i \in [m]$ and all variables $x_j \in \CV$ where $\CV$ is an infinite set of variables.
We write $\varphi(x_1,\dots,x_k)$ to indicate that the free variables of $\varphi$ are among the variables $\{x_1,\dots,x_k\}$.
For a structure $\FA = (V(\FA),R_1^{\FA},\dots,R_m^{\FA})$ and $\vec v = (v_1,\dots,v_k) \in (V(\FA))^{k}$ we write $\FA \models \varphi(\vec v)$ if $\FA$ is a model of $\varphi$ when $x_i$ is interpreted by $v_i$.

We define the \emph{quantifier rank} of a formula $\varphi \in \FO$ inductively via
\begin{itemize}
 \item $\qr(x_1 = x_2) = \qr(R_i(x_1,\dots,x_{k_i})) \coloneqq 0$ for all  $i \in [m]$ and all variables $x_j \in \CV$,
 \item $\qr(\varphi \wedge \psi) \coloneqq \max(\qr(\varphi),\qr(\psi))$,
 \item $\qr(\neg \varphi) \coloneqq \qr(\varphi)$, and
 \item $\qr(\exists x \varphi) \coloneqq \qr(\varphi) + 1$ for all $x \in \CV$.
\end{itemize}

We define \emph{first-order logic with counting quantifiers} $\LC$ to be the extension of $\FO$ by counting quantifiers of the form $\exists^{\geq j} x \varphi$.
The formula $\exists^{\geq j} x \varphi$ is satisfied over a structure $\FA$ if there are at least $j$ distinct elements $v \in V(\FA)$ that satisfy $\varphi$.
We extend the definition of the quantifier rank in the natural way by setting $\qr(\exists^{\geq j} x \varphi) \coloneqq \qr(\varphi) + 1$ for all $x \in \CV$.

For $k \in \NN$ we define $\LLk{k}$ to be the restriction of $\FO$ to formulas over at most $k$ variables, i.e., we restrict ourselves to a set of variables $\CV$ of size exactly $k$.
Similarly, we define $\LCk{k}$ to be the restriction of $\LC$ to formulas over at most $k$ variables.

Moreover, for $q \geq 0$, we define $\LLkq{k}{q}$ to the restriction of $\LLk{k}$ to formulas $\varphi$ of quantifier rank $\qr(\varphi) \leq q$.
Similarly, we define $\LCkq{k}{q}$ to the restriction of $\LCk{k}$ to formulas of quantifier rank at most $q$.

\paragraph{The Weisfeiler-Leman Algorithm.}

Next, we describe the $k$-WL algorithm.
While it is most commonly used as a heuristic to graph isomorphism testing, the algorithm can be applied to any relational structure of arity at most $k$.

Let $\chi_1,\chi_2 \colon V^{k} \rightarrow C$ be colorings of $k$-tuples over a finite set $V$ where $C$ is some finite set of colors.
The coloring $\chi_1$ \emph{refines} $\chi_2$, denoted $\chi_1 \preceq \chi_2$, if $\chi_1(\vec v) = \chi_1(\vec w)$ implies $\chi_2(\vec v) = \chi_2(\vec w)$ for all $\vec v,\vec w \in V^{k}$.
Observe that $\chi_1 \preceq \chi_2$ if and only if the partition into color classes of $\chi_1$ refines the corresponding partition into color classes of $\chi_2$.
The colorings $\chi_1$ and $\chi_2$ are \emph{equivalent}, denoted $\chi_1 \equiv \chi_2$, if $\chi_1 \preceq \chi_2$ and $\chi_2 \preceq \chi_1$.
Also, $\chi_1$ \emph{strictly refines} $\chi_2$, denoted $\chi_1 \prec \chi_2$, if $\chi_1 \preceq \chi_2$ and $\chi_1 \not\equiv \chi_2$.

Let us fix $k \geq 2$ and consider a relational structure $\FA = (V(\FA),R_1^{\FA},\dots,R_m^{\FA})$ of arity at most $k$.
Let $\vec v = (v_1,\dots,v_k) \in (V(\FA))^{k}$.
We define the \emph{atomic type} of $\vec v$, denoted by $\atp_\FA(\vec v)$, to be the isomorphism type of the ordered substructure of $\FA$ that is induced by $\{v_1,\dots,v_k\}$.
More concretely, for a second structure $\FB = (V(\FB),R_1^{\FB},\dots,R_m^{\FB})$ and a tuple $\vec w = (w_1,\dots,w_k) \in (V(\FB))^{k}$, it holds that $\atp_\FA(\vec v) = \atp_\FB(\vec w)$ if the mapping $v_i \mapsto w_i$ is an isomorphism from $\FA[\{v_1,\dots,v_k\}]$ to $\FB[\{w_1,\dots,w_k\}]$.

Next, we describe a single refinement step of $k$-WL.
Let $V$ be a finite set and let $\chi\colon V^{k} \rightarrow C$ be a coloring of all $k$-tuples over $V$.
We define the coloring $\refWL{k}{\chi}$ by setting
\[\big(\refWL{k}{\chi}\big)(\vec v) \coloneqq \Big(\chi(\vec v),\CM_\chi(\vec v)\Big)\]
for all $\vec v = (v_1,\dots,v_k) \in V^{k}$ where
\[\CM_\chi(\vec v) \coloneqq \Big\{\!\!\Big\{\big(\chi(\vec v[w/1]),\dots,\chi(\vec v[w/k])\big) \Bigmid w \in V\Big\}\!\!\Big\}\]
and $\vec v[w/i] \coloneqq (v_1,\dots,v_{i-1},w,v_{i+1},\dots,v_k)$ is the tuple obtained from $\vec v$ by replacing the $i$-th entry by $w$ (and $\{\!\{\dots\}\!\}$ denotes a multiset).
Observe that $\refWL{k}{\chi} \preceq \chi$.
We say the coloring $\chi$ is \emph{$k$-stable} if $\refWL{k}{\chi} \equiv \chi$.

We define the initial coloring computed by $k$-WL on the structure $\FA$ via $\WLit{k}{0}{\FA}(\vec v) \coloneqq \atp_\FA(\vec v)$ for all $\vec v \in (V(\FA))^{k}$.
For $r \geq 0$ we set
\[\WLit{k}{r+1}{\FA} \coloneqq \refWL{k}{\WLit{k}{r}{\FA}}.\]
Since $\WLit{k}{r+1}{\FA} \preceq \WLit{k}{r}{\FA}$ for all $r \geq 0$, there is some minimal $r_{\infty} \leq |V|^{k} - 1$ such that
\[\WLit{k}{r_{\infty}}{\FA} \equiv \WLit{k}{r_{\infty}+1}{\FA}.\]
We say that $k$-WL \emph{stabilizes after $r_\infty$ rounds on $\FA$} and define $\WL{k}{\FA} \coloneqq \WLit{k}{r_{\infty}}{\FA}$ to be the output coloring of $k$-WL.
Observe that $\WL{k}{\FA}$ is a $k$-stable coloring.

Now, let $\FB = (V(\FB),R_1^{\FB},\dots,R_m^{\FB})$ be a second structure.
Let $r \geq 0$.
We say that $k$-WL \emph{distinguishes $\FA$ and $\FB$ after $r$ rounds} if there is some color $c$ such that
\[\Big|\Big\{\vec v \in (V(\FA))^{k} \Bigmid \WLit{k}{r}{\FA}(\vec v) = c\Big\}\Big| \neq \Big|\Big\{\vec w \in (V(\FB))^{k} \Bigmid \WLit{k}{r}{\FB}(\vec w) = c\Big\}\Big|.\]
We also say that $k$-WL \emph{distinguishes $\FA$ and $\FB$} if there is some integer $r \geq 0$ such that $k$-WL distinguishes $\FA$ and $\FB$ after $r$ rounds.
We write $\FA \simeq_k \FB$ if $k$-WL does not distinguish $\FA$ and $\FB$.
Note that, if $k$-WL distinguishes $\FA$ and $\FB$ and $k$-WL stabilizes after $r_\infty$ rounds on $\FA$, then $k$-WL distinguishes $\FA$ and $\FB$ after (at most) $r_\infty + 1$ rounds.

The following connections to bounded-variable fragments of first-order logic with counting quantifiers are well-known.
Those connections were first proved in \cite{CaiFI92,ImmermanL90} for graphs, but the arguments directly generalize to arbitrary relational structures (see, e.g., \cite{Grohe17}).

\begin{theorem}
 \label{thm:wl-logic}
 Let $k \geq 2$.
 Also let $\FA$ and $\FB$ be structures of arity at most $k$ and suppose $\vec v \in V(\FA)^k$ and $\vec w \in V(\FB)^k$.
 Then, for every $r \geq 0$, it holds that $\WLit{k}{r}{\FA}(\vec v) \neq \WLit{k}{r}{\FB}(\vec w)$ if and only if there is some $\varphi(\vec x) \in \LCkq{k+1}{r}$ such that $\FA \models \varphi(\vec v)$ and $\FB \not\models \varphi(\vec w)$. 
\end{theorem}

\begin{corollary}
 \label{cor:wl-logic-distinguish-structures}
 Let $k \geq 2$.
 Also let $\FA$ and $\FB$ be structures of arity at most $k$.
 
 If there is a sentence $\varphi \in \LCkq{k+1}{r}$ such that $\FA \models \varphi$ and $\FB \not\models \varphi$, then the $k$-dimensional Weisfeiler-Leman algorithm distinguishes $\FA$ and $\FB$ after at most $r$ refinement rounds.
 
 If the $k$-dimensional Weisfeiler-Leman algorithm distinguishes $\FA$ and $\FB$ after $r$ refinement rounds, then there is a sentence $\varphi \in \LCkq{k+1}{r+k}$ such that $\FA \models \varphi$ and $\FB \not\models \varphi$.
\end{corollary}

\paragraph{Algebras.}

Finally, we recall some algebraic tools required in this work.
We use $\Complex$ to denote the complex numbers.

Recall that a \emph{$\Complex$-algebra} $\Alg$ is a ring which is also a $\Complex$-vector space such that $a \cdot (\vec v \vec w) = (a \cdot \vec v) \vec w =\vec v (a \cdot \vec w)$ for all $a \in \Complex$ and $\vec v,\vec w \in \Alg$.
Since we restrict our attention to complex numbers, we simply refer to a $\Complex$-algebra as an algebra.
In this work, we are interested in matrix algebras where the algebra consists of $(d\times d)$-matrices over the complex numbers with standard matrix multiplication as the ring operation.
We write $\FullMatrixAlg_d(\Complex)$ for the full matrix algebra of all $(d \times d)$-matrices over the complex numbers.
It is a well-known fact that a matrix algebra $\Alg \subseteq  \FullMatrixAlg_d(\Complex)$, which is closed under conjugate transposition, is always \emph{semisimple}.
Indeed, if $M$ is in the Jacobson radical of $\Alg$, then so is $M^*M$.
But $M^*M$ is diagonalizable (because it is Hermitian) and nilpotent (because the radical is nilpotent \cite[Lemma 1.6.6]{Zimmermann14}) and hence, $M^*M = 0$ and so $M = 0$.
Then the radical itself is $0$, which is one characterization of semisimplicity.

Hence, we can use the following result to bound the length of sequences of strict subalgebras of $\FullMatrixAlg_d(\Complex)$ that are closed under conjugate transposition.

\begin{theorem}[{\cite[Theorem~5]{LichterPS19}}]
 \label{thm:semisimple-sequence-length}
 Let $\Alg^{(1)} \subset \dots \subset \Alg^{(\ell)} \subseteq \FullMatrixAlg_d(\Complex)$ be a sequence of semisimple strict subalgebras.
 Then $\ell \leq 2d$.
\end{theorem}

A \emph{$*$-algebra} is an algebra with an additional operation $*$ such that $(\vec v + \vec w)^* = \vec v^* + \vec w^*$, $(\vec v \vec w)^* = \vec w^* \vec v^* $, $\vec 1^* = \vec 1$ and $(\vec v^*)^* = \vec v$ for all $\vec v,\vec w \in \Alg$ (where $\vec 1$ denotes the unit element).
Note that $\FullMatrixAlg_d(\Complex)$ forms a $*$-algebra using conjugate transposition.

%% file: upper-bounds.tex
In this section, we prove Theorem \ref{thm:wl-round-upper-bound}.
Actually, we prove a more general result on the maximal iteration number of any refinement method that is at least as strong as $k$-WL.

For the remainder of this section, let us fix some integer $k \geq 2$.
Let $V$ be a finite set and let $\CP$ be a partition of $V^k$.
For two tuples $\vec v,\vec w \in V^k$, we write $\vec v \sim_\CP \vec w$ if there is some $P \in \CP$ such that $\vec v,\vec w \in \CP$ (i.e., $\sim_\CP$ is the equivalence relation with equivalence classes from $\CP$).

We say $\CP$ is \emph{compatible with equality} if for all $P \in \CP$, all tuples $(v_1,\ldots,v_k),(w_1,\ldots,w_k) \in P$,
and all $i,j \in [k]$ it holds that
\[v_i = v_j \iff w_i = w_j.\]
Moreover, the partition $\CP$ is \emph{shufflable} if for every function $\pi\colon [k] \rightarrow [k]$ and every pair of tuples $(v_1,\dots,v_k),(w_1,\dots,w_k) \in V^k$ it holds that
\begin{equation}
 \label{eq:shuffle}
 (v_1,\dots,v_k) \sim_\CP (w_1,\dots,w_k) \quad\implies\quad (v_{\pi(1)},\dots,v_{\pi(k)}) \sim_\CP (w_{\pi(1)},\dots,w_{\pi(k)}).
\end{equation}

\begin{observation}
 \label{obs:shuffle-bijection}
 Let $\CP$ be a shufflable partition of $V^k$.
 Then
 \[P^{\pi}\coloneqq \left\{(v_{\pi(1)},\ldots,v_{\pi(k)}) \;\middle|\; (v_1,\ldots,v_k)\in P\right\}\in\CP\]
 for every bijection $\pi\colon [k] \rightarrow [k]$ and every $P \in \CP$.
\end{observation}

\begin{proof}
 Let $Q \in \CP$ such that $Q \cap P^{\pi} \neq \emptyset$.
 This means there is some $(v_1,\dots,v_k) \in P$ such that $(v_{\pi(1)},\dots,v_{\pi(k)}) \in Q$.
 Let $(w_1,\dots,w_k) \in P$ be another tuple.
 Then $(v_1,\dots,v_k) \sim_\CP (w_1,\dots,w_k)$ and thus, $(v_{\pi(1)},\dots,v_{\pi(k)}) \sim_\CP (w_{\pi(1)},\dots,w_{\pi(k)})$ by Equation \eqref{eq:shuffle}.
 Since $(v_{\pi(1)},\dots,v_{\pi(k)}) \in Q$, it follows that $(w_{\pi(1)},\dots,w_{\pi(k)}) \in Q$.
 So $P^{\pi} \subseteq Q$.

 By the same argument, $Q^{\pi^{-1}} \subseteq P$ which implies that $Q \subseteq P^{\pi}$.
 Together, this means that $P^{\pi} = Q \in \CP$.
\end{proof}

We say a coloring $\chi\colon V^k \rightarrow C$ of $k$-tuples is \emph{compatible with equality} if the corresponding partition $\CP$ into color classes is compatible with equality.
Similarly, $\chi$ is \emph{shufflable} if $\CP$ is shufflable.

Recall that $\refWL{k}{\chi}$ denotes the coloring obtained from $\chi$ after applying a single refinement round of $k$-WL.

\begin{theorem}
 \label{thm:upper-bound}
 Let $V$ be a finite set of size $n \coloneqq |V|$.
 Also let $\chi_0,\dots,\chi_\ell \colon V^k \rightarrow C$ be a sequence of colorings such that
 \begin{enumerate}[label = (\Roman*)]
  \item\label{item:upper-bound-shufflable} $\chi_t$ is shufflable and compatible with equality for all $t \in [0,\ell]$,
  \item\label{item:upper-bound-wl} $\refWL{k}{\chi_{t-1}} \succeq \chi_t$ for all $t \in [\ell]$, and
  \item\label{item:upper-bound-strict} $\chi_{t-1} \succ \chi_t$ for all $t \in [\ell]$.
 \end{enumerate}
 Then $\ell \leq 2n^{k-1}(\ceil{k\log n} + 1) = O(kn^{k-1}\log n)$.
\end{theorem}

Note that Theorem \ref{thm:wl-round-upper-bound} immediately follows from Theorem \ref{thm:upper-bound} by observing that all colorings $\WLit{k}{i}{\FA}$ obtained from the refinement process of $k$-WL are shufflable and compatible with equality.

The proof of Theorem \ref{thm:upper-bound} relies on algebraic tools.
Let $V$ be a finite set of size $n \coloneqq |V|$.
We define a multiplication on the space $\Complex^{V^k}$ by 
\begin{equation}
 (\vec a\cdot\vec b)(v_1,\ldots,v_k)\coloneqq\sum_{v\in V}\vec a(v_1,\ldots,v_{k-1},v)\vec b(v_1,\ldots,v_{k-2},v,v_k)
\end{equation}
for all $\vec a,\vec b\in\Complex^{V^k}$.
Note that this multiplication is associative and has a unit $\vec 1$, defined by
\[\vec 1(v_1,\ldots,v_k) \coloneqq \begin{cases}
                                    1&\text{if }v_{k-1}=v_k,\\
                                    0&\text{otherwise.}
                                   \end{cases}
\]
Furthermore, the multiplication is compatible with the vector space structure.
Hence, it defines an algebra which we denote by $\Alg$.

With every $\vec a \in \Complex^{V^k}$ we associate a matrix $M_{\vec a}\in\Complex^{V^{k-1}\times V^{k-1}}$ with entries
\begin{align*}
 &M_{\vec a}\big((v_1,\dots,v_{k-1}),(w_1,\dots,w_{k-1})\big) \coloneqq\\
 &\quad\quad
 \begin{cases}
  \vec a(v_1,\dots,v_{k-2},v_{k-1},w_{k-1}) &\text{if }v_i=w_i\text{ for all }i\in[k-2],\\
  0                                         &\text{otherwise}.
 \end{cases}
\end{align*}

It is easy to see that the mapping $\vec a\mapsto M_{\vec a}$ is injective and linear.
Moreover, it is compatible with multiplication:
\begin{align*}
 &M_{\vec a}\cdot M_{\vec b}\big((v_1,\ldots,v_{k-1}),(w_1,\ldots,w_{k-1})\big)\\
 &= \sum_{u_1,\ldots,u_{k-1}\in V}M_{\vec a}\big((v_1,\ldots,v_{k-1}),(u_1,\ldots,u_{k-1})\big) M_{\vec b}\big((u_1,\ldots,u_{k-1}),(w_1,\ldots,w_{k-1})\big)\\
 &= \begin{cases}
     \sum_{u\in V}\vec a(v_1,\ldots,v_{k-2},v_{k-1},u) \vec b(v_1,\ldots,v_{k-2},u,w_{k-1}) &\text{if } v_i=w_i\text{ for all }i\in[k-2],\\
     0                                                                                      &\text{otherwise}
    \end{cases}\\
 &= \begin{cases}
     (\vec a\cdot\vec b)(v_1,\ldots,v_{k-1},w_{k-1}) &\text{if } v_i=w_i\text{ for all }i\in[k-2],\\
     0                                               &\text{otherwise}
    \end{cases}\\
 &= M_{\vec a\cdot\vec b}\big((v_1,\ldots,v_{k-1}),(w_1,\ldots,w_{k-1})\big).
\end{align*}
And finally, $M_{\vec 1}$ is the identity matrix.
Thus, $\Alg$ is isomorphic to a subalgebra of the $n^{k-1}\times n^{k-1}$-dimensional matrix algebra $\Complex^{V^{k-1}\times V^{k-1}}$.

For every $\vec a\in\Alg$ we define $\vec a^*\in\Alg$ by
\[\vec a^*(v_1,\ldots,v_k) \coloneqq \overline{\vec a(v_1,\ldots,v_{k-2},v_k,v_{k-1})}\]
(here, $\overline{c}$ denotes the complex conjugate of a number $c \in \Complex$, i.e., if $c = a + bi$ then $\overline{c} = a - bi$).
Then $M_{\vec a^*} = (M_{\vec a})^*$ (the conjugate transpose).
Thus, $^*$ is an involution on $\Alg$ compatible with the algebra structure, which turns $\Alg$ into a $*$-algebra.

Since $\Alg$ is isomorphic to a subalgebra of $\FullMatrixAlg_{n^{k-1}}(\Complex)$ which is closed under conjugate transposition, we conclude that $\Alg$ is semisimple.
Moreover, Theorem \ref{thm:semisimple-sequence-length} implies the following corollary.

\begin{corollary}
 \label{cor:semisimple-length}
 Let $\Alg^{(1)} \subset \dots \subset \Alg^{(\ell)} \subseteq \Alg$ be a sequence of semisimple strict subalgebras of $\Alg$.
 Then $\ell \leq 2n^{k-1}$.
\end{corollary}

We wish to use the last corollary to obtain an upper bound on the length of the coloring sequence $\chi_0,\dots,\chi_\ell$ in Theorem \ref{thm:upper-bound}.
Towards this end, we associate with every coloring $\chi_t$ (or the corresponding partition into color classes) a subalgebra of $\Alg$ as follows.

For every subset $A\subseteq \Complex^{V^k}$, we let $\spn(A)$ be the linear subspace of $\Complex^{V^k}$ generated by $A$,
and we let $\angles{A}$ be the closure of $\spn(A)$ under multiplication.
If $\vec 1\in\angles{A}$, then $\angles{A}$ is a subalgebra of $\Alg$.
As indicated above, we are interested in subalgebras of $\Alg$ generated by partitions of the set $V^k$ in the way explained next.

For every subset $P\subseteq V^k$, we define
\[
 \vec c_P(\vec v)\coloneqq
 \begin{cases}
  1 &\text{if }\vec v\in P,\\
  0 &\text{otherwise}
 \end{cases}
\]
to be the characteristic vector of $P$.
For a partition $\CP$ of $V^k$, we let $C_{\CP}\coloneqq\{\vec c_P\mid P\in\CP\}$ and $\Alg_{\CP}\coloneqq\angles{C_{\CP}}$.
If $\vec 1\in \Alg_{\CP}$, then $\Alg_\CP$ is a subalgebra of $\Alg$.

\begin{lemma}
 \label{la:partiton-inclusion-to-algebra}
 Let $\CP$ and $\CQ$ be partitions of $V^k$ such that $\CQ$ strictly refines $\CP$.
 Then $\spn(C_{\CP})\subset\spn(C_{\CQ})$ and $\Alg_{\CP}\subseteq\Alg_{\CQ}$.
\end{lemma}

\begin{proof}
 If $P \in \CP$ is the disjoint union of $Q_1,\ldots,Q_m \in \CQ$, then $\vec c_P = \sum_{i=1}^m \vec c_{Q_i}$.
 Thus, $C_{\CP} \subseteq \spn(C_{\CQ})$ and therefore $\spn(C_{\CP})\subseteq\spn(C_{\CQ})$.
 Moreover, there are $P\in\CP,Q\in\CQ$ such that $Q \subset P$.
 Then $\vec c_Q\not\in\spn(C_{\CP})$, because all $\vec a\in \spn(C_{\CP})$ are constant on $P$.
 Hence the inclusion is strict.
 
 The second assertion $\Alg_{\CP}\subseteq\Alg_{\CQ}$ follows immediately from the definitions of $\Alg_{\CP}$ and $\Alg_{\CQ}$.
\end{proof}

\begin{observation}
 \label{obs:partition-properties}
 Let $\CP$ be a partition of $V^k$.
 \begin{enumerate}[label = (\arabic*)]
  \item If $\CP$ is compatible with equality, then $\vec 1\in\Alg_\CP$ and hence $\Alg_{\CP}$ is a subalgebra of $\Alg$.
  \item If $\CP$ is shufflable, then $\Alg_\CP$ is closed under $^*$.
 \end{enumerate}
\end{observation}

\begin{proof}
 Suppose that $\CP$ is compatible with equality.
 Then there is some $\CQ \subseteq \CP$ such that $\vec 1 = \sum_{P \in \CQ} \vec c_P$.
 Hence, $\vec 1 \in \Alg_\CP$.

 Next, suppose that $\CP$ is shufflable.
 Consider the bijection $\pi\colon [k] \rightarrow [k]$ for which $\pi(i) = i$ for all $i \in [k-2]$, $\pi(k-1) = k$ and $\pi(k) = k-1$.
 Then $\vec c_P^* = \vec c_{P^\pi}$ for every $P \in \CP$.
 Using Observation \ref{obs:shuffle-bijection}, it follows that $\vec c_P^* \in C_\CP$ which implies that $\Alg_\CP$ is closed under $^*$.
\end{proof}

\begin{corollary}
 \label{cor:partition-algebra-semisimple}
 Let $\CP$ be a partition of $V^k$ that is shufflable and compatible with equality.
 Then $\Alg_\CP$ is a $*$-subalgebra of $\Alg$.
 In particular, $\Alg_\CP$ is semisimple.
\end{corollary}

Recall that our goal is to bound the length of the color sequence $\chi_0,\dots,\chi_\ell$ in Theorem \ref{thm:upper-bound}.
We associate a $*$-subalgebra $\Alg^{(t)}$ of $\Alg$ with every coloring $\chi_t$ by considering the corresponding partition into color classes.
The last corollary implies that $\Alg^{(t)}$ is semisimple for every $t \in [0,\ell]$.
So, to be able to apply Corollary \ref{cor:semisimple-length} to bound the length of the sequence of subalgebras, it remains to argue that inclusions between successive subalgebras are strict.
Actually, this is not true in general, but we can prove that only a small number of successive algebras can be equal.

We say that $\vec a\in\Complex^{V^k}$ \emph{distinguishes} $\vec v,\vec w\in V^{k}$ if $\vec a(\vec v) \neq \vec a(\vec w)$,
and we say that $A\subseteq\Complex^{V^k}$ distinguishes $\vec v,\vec w$ if some $\vec a\in A$ distinguishes them.

\begin{lemma}
 \label{la:multiplication-length}
 Let $A\subseteq\Complex^{V^k}$ and $\vec v,\vec w \in V^{k}$ such that $\angles{A}$ distinguishes $\vec v,\vec w$.
 Then there are an $s \le n^{k}$ and $\vec a_1,\ldots,\vec a_s\in A$ such that $\vec a_1\cdots\vec a_s$ distinguishes $\vec v,\vec w$.
\end{lemma}

\begin{proof}
 As a linear subspace of $\Complex^{V^k}$, the space $\angles{A}$ consists of finite linear combinations of ``monomials'' $\vec a_1\cdots\vec a_s$ for $\vec a_i\in A$.
 Since the dimension of the space is at most $n^k$, we only need to consider such monomials for $s \le n^k$.
 Hence $\vec v,\vec w$ are distinguished by a linear combination
 \[\sum_{i=1}^m\lambda_i\vec a_{i1}\cdots\vec a_{is_i}\]
 with $\lambda_i\in\Complex$, $\vec a_{ij}\in A$, and $s_i \le n^k$.
 This immediately implies that $\vec v,\vec w$ are distinguished by $\vec a_{i1}\cdots\vec a_{is_i}$ for some $i \in [m]$.
\end{proof}

With every partition $\CP=\{P_1,\ldots,P_m\}$ we associate a relational structure $(V,R_1^{\CP},\ldots,R_m^{\CP})$ whose vocabulary consists of $k$-ary relation symbols $R_i$ interpreted by $R_i^{\CP}=P_i$
(to uniquely define the associated structure, we fix an arbitrary order on the blocks $P_1,\dots,P_m$).
Slightly abusing notation, we denote this structure by $\CP$ as well.
We say that a formula $\varphi(\vec x)$ \emph{distinguishes $\vec v,\vec w\in V^{k}$ over $\CP$} if
\[\CP\models\varphi(\vec v)\quad\iff\quad\CP\not\models\varphi(\vec w).\]
Recall that $\LCkq{k+1}{q}$ denotes the fragment of first-order logic with counting consisting of all formulas of quantifier rank at most $q$ with at most $k+1$ variables.

\begin{lemma}
 \label{la:multiplication-in-logic}
 Let $\CP$ be a partition of $V^k$ and let $\vec v,\vec w \in V^{k}$ such that $\Alg_{\CP}$ distinguishes $\vec v$ and $\vec w$.
 Then there is a formula $\varphi(\vec x)\in\LCkq{k+1}{q}$ of quantifier rank $q\le \ceil{k\log n}$ that distinguishes $\vec v,\vec w$ over $\CP$.
\end{lemma}

\begin{proof}
 Suppose that $\CP=\{P_1,\ldots,P_m\}$, and let $\vec c_i\coloneqq\vec c_{P_i}$.
 Then $\Alg_{\CP}=\angles{\{\vec c_1,\ldots,\vec c_m\}}$.
 Thus, by Lemma~\ref{la:multiplication-length}, there is an $s\le n^k$ and $i_1,\ldots,i_s\in[m]$ such that $\vec c_{i_1}\cdots \vec c_{i_s}$ distinguishes $\vec v,\vec w$.
 
 By induction on $s\ge 1$, we prove that if $\vec c_{i_1}\cdots \vec c_{i_s}$ distinguishes $\vec v,\vec w$,
 then there is a formula $\varphi(\vec x)\in\LCkq{k+1}{\ceil{\log s}}$ that distinguishes $\vec v,\vec w$.
 The assertion of the lemma follows.
 
 For the base step $s=1$, note that if $\vec c_i$ distinguishes $\vec v,\vec w$, then the atomic formula $R_i(\vec x)$ distinguishes $\vec v,\vec w$.

 For the inductive step, let $s \ge 2$.
 Suppose that $\vec b = \vec c_{i_1}\cdots \vec c_{i_s}$ distinguishes $\vec v,\vec w$.
 Let $r\coloneqq\ceil{s/2}$ and note that $r\le 2^{\ceil{\log s}-1}$ and therefore
 \[\ceil{\log r}\le \ceil{\log s}-1.\]
 Let $\vec b_1\coloneqq \vec c_{i_1}\cdots \vec c_{i_r}$ and $\vec b_2\coloneqq \vec c_{i_{r+1}}\cdots \vec c_{i_s}$.
 Then $\vec b = \vec b_1\cdot \vec b_2$.
 Suppose that $\vec v=(v_1,\ldots,v_k)$ and $\vec w=(w_1,\ldots,w_k)$.
 We have
  \begin{align*}
    \vec b(\vec v)
    &=\sum_{u\in V}
      \vec b_1(v_1,\ldots,v_{k-1},u) \cdot \vec b_2(v_1,\ldots,v_{k-2},u,v_k)\\            
    \neq\;\vec b(\vec w)
    &=\sum_{u\in V}
      \vec b_1(w_1,\ldots,w_{k-1},u) \cdot \vec b_2(w_1,\ldots,w_{k-2},u,w_k).
  \end{align*}
  Thus, there are $b_1,b_2\in\Complex$ such that
  \begin{align*}
    p\coloneqq\;&\Big|\Big\{u\in V\Bigmid
      \vec b_1(v_1,\ldots,v_{k-1},u)=b_1\text{ and
      }\vec b_2(v_1,\ldots,v_{k-2},u,v_k)=b_2\Big\}\Big|\\
    \neq\;&\Big|\Big\{u\in V\Bigmid
            \vec b_1(w_1,\ldots,w_{k-1},u)=b_1\text{ and
            }\vec b_2(w_1,\ldots,w_{k-2},u,w_k)=b_2\Big\}\Big|            \eqqcolon q.
  \end{align*}
  It follows from the induction hypothesis that for $i=1,2$ and for
  all $\vec v',\vec w'\in V^k$ such that $\vec b_i$ distinguishes $\vec
  v',\vec w'$ there is a formula $\psi_i^{\vec v',\vec
    w'}(\vec x) \in\LCkq{k+1}{\ceil{\log r}}$ that distinguishes $\vec v',\vec w'$. Without loss of
  generality,
  \[
    \CP\models \psi_i^{\vec v',\vec
      w'}(\vec v')\quad\text{and}\quad \CP\not\models \psi_i^{\vec v',\vec
      w'}(\vec w'),
  \]
  otherwise we replace $\psi_i^{\vec v',\vec
    w'}(\vec x)$ by its negation. Let $V_i\subseteq V^k$ be the set of
  all $\vec v'\in V^k$ such that $\vec b_i(\vec v')=b_i$ and let
  \[
    \varphi_i(\vec x)\coloneqq\bigvee_{\vec v'\in V_i}\bigwedge_{\vec w'\in V^k\setminus V_i}
    \psi_i^{\vec v',\vec
      w'}(\vec x).
  \]
  Then for all $\vec v'\in V^k$ we have
  \[
    \CP\models \varphi_i(\vec v')\iff
    \vec b_i(\vec v')=b_i.
  \]
  Without loss of generality we assume that $p>q$. Then the formula
  \[
     \varphi(x_1,\ldots,x_k)\coloneqq\exists^{\ge
       p}x_{k+1}\big(\varphi_1(x_1,\ldots,x_{k-1},x_{k+1})\wedge
     \varphi_2(x_1,\ldots,x_{k-2},x_{k+1},x_k)\big)\in
   \LCkq{k+1}{\ceil{\log s}}
  \]
  distinguishes $\vec v,\vec w$.
\end{proof}

We are now ready to prove Theorem \ref{thm:upper-bound}.

\begin{proof}[Proof of Theorem \ref{thm:upper-bound}]
 For every $t \in [0,\ell]$ let $\CP^{(t)}$ be the partition of $V^k$ into the color classes of $\chi_t$.
 
 \begin{claim}
  \label{claim:multiplication-logic}
  Let $t,q \geq 0$ such that $t+q \leq \ell$.
  Suppose that there is a formula $\varphi(\vec x)\in \LCkq{k+1}{q}$ that distinguishes $\vec v,\vec w\in V^k$ over $\CP^{(t)}$.
  Then $\vec v,\vec w$ belong to different classes of the partition $\CP^{(t+q)}$.
 \end{claim}

 \begin{claimproof}
  By Condition \ref{item:upper-bound-shufflable}, the partition $\CP^{(t)}$ is shufflable and compatible with equality.
  This implies that $\chi_t \equiv \WLit{k}{0}{\CP^{(t)}}$.
  Together with Condition \ref{item:upper-bound-wl}, we get that $\chi_{t + q} \preceq \WLit{k}{q}{\CP^{(t)}}$.
  
  Also, using Theorem \ref{thm:wl-logic}, we get that $\WLit{k}{q}{\CP^{(t)}}(\vec v) \neq \WLit{k}{q}{\CP^{(t)}}(\vec w)$.
  Overall, it follows that $\vec v,\vec w$ belong to different classes of the partition $\CP^{(t+q)}$.
 \end{claimproof}
 
 For every $t \in [0,\ell]$ we define $C^{(t)} \coloneqq C_{\CP^{(t)}}$ and $\Alg^{(t)} \coloneqq \Alg_{\CP^{(t)}}$.
 Note that $\Alg^{(t)}$ is a semisimple $*$-subalgebra of $\Alg$ by Condition \ref{item:upper-bound-shufflable} and Corollary \ref{cor:partition-algebra-semisimple}.
 By Lemma~\ref{la:partiton-inclusion-to-algebra}, we have
 \begin{equation}
  \label{eq:algebra-sequence}
  \Alg^{(0)}\subseteq\Alg^{(1)}\subseteq\ldots\subseteq\Alg^{(\ell)}\subseteq\Alg.
 \end{equation}
 
 \begin{claim}
  \label{claim:algebra-inclusion}
  For all $t \in [0,\ell - \ceil{k\log n}]$,
  \[\Alg^{(t)}\subseteq\spn(C^{(t+\ceil{k\log n})}).\]
 \end{claim}
 
 \begin{claimproof}
  Let $\vec a\in\Alg^{(t)}$.
  By Lemma~\ref{la:multiplication-in-logic} and Claim~\ref{claim:multiplication-logic},
  for all $\vec v,\vec w\in V^k$, if $\vec a(\vec v)\neq\vec a(\vec w)$, that is, if $\vec a$ distinguishes $\vec v$ and $\vec w$,
  then $\vec v$ and $\vec w$ belong to different classes of the partition $\CP^{(t+\ceil{k\log n})}$.
  Thus, $\vec a$ is constant on each class of the partition $\CP^{(t+\ceil{k\log n})}$, which immediately implies that $\vec a$ can be written as a linear combination of the characteristic vectors $\vec c_P$ of the classes $P\in \CP^{(t+\ceil{k\log n})}$.
  This is the assertion of the claim.
 \end{claimproof}
 
 \begin{claim}
  \label{claim:algebra-inclusion-strict}
  For all $t \in [0,\ell-\ceil{k\log n} - 1]$,
  \[\Alg^{(t)} \subset \Alg^{(t+\ceil{k\log n}+1)}.\]
 \end{claim}
 
 \begin{claimproof}
  By Claim~\ref{claim:algebra-inclusion}, we have $\Alg^{(t)}\subseteq\spn(C^{(t+\ceil{k\log n})})$.
  Moreover, by Condition \ref{item:upper-bound-strict}, the partition $\CP^{(t+\ceil{k\log n}+1)}$ strictly refines the partition $\CP^{(t+\ceil{k\log n})}$.
  By Lemma~\ref{la:partiton-inclusion-to-algebra}, this implies
  $\spn(C^{(t+\ceil{k\log n})})\subset\spn(C^{(t+\ceil{k\log n}+1)})$.
  As $\spn(C^{(t+\ceil{k\log n}+1)})\subseteq\Alg^{(t+\ceil{k\log n}+1)}$, the assertion of the claim follows.
 \end{claimproof}
 
 Recall that by Corollary~\ref{cor:partition-algebra-semisimple}, the algebras $\Alg^{(t)}$ are semisimple.
 Thus, by Corollary \ref{cor:semisimple-length}, at most $2n^{k-1}$ of the inclusions in \eqref{eq:algebra-sequence} are strict.
 Then Claim~\ref{claim:algebra-inclusion-strict} implies
 \[\ell \leq 2n^{k-1}(\ceil{k\log n}+1) = O(kn^{k-1}\log n).\qedhere\]
\end{proof}

%% file: long-sequences.tex
Next, we prove an almost matching lower bound for Theorem \ref{thm:upper-bound}, i.e., we prove that there are sequences of colorings $\chi_0,\dots,\chi_\ell \colon V^k \rightarrow C$ satisfying Conditions \ref{item:upper-bound-shufflable} - \ref{item:upper-bound-strict} of Theorem \ref{thm:upper-bound} of length $\ell = \Omega(n^{k-1})$.
Actually, we prove a slightly stronger result.

As before, let us fix an integer $k \geq 2$.
We present a construction for a sequence $\chi_0 \succ \chi_1 \succ \dots \succ \chi_\ell$ of colorings of $V^{k}$ such that $\chi_t$ is $k$-stable (i.e., the coloring is stable with respect to $k$-WL) for all $t \in [0,\ell]$.
More precisely, the main result of this section is the following theorem.

\begin{theorem}
 \label{thm:long-sequence-stable}
 Suppose $n \geq 2k^{2}$ and let $V$ be a set of size $|V| = 2n$.
 Then there is a sequence of colorings $\chi_0,\dots,\chi_\ell \colon V^k \rightarrow C$ of length $\ell \geq \left(\frac{n}{2k}\right)^{k-1}$ such that
 \begin{enumerate}[label = (\Roman*)]
  \item $\chi_t$ is shufflable and compatible with equality for all $t \in [0,\ell]$,
  \item $\chi_t$ is $k$-stable for all $t \in [0,\ell]$, and
  \item $\chi_{t-1} \succ \chi_t$ for all $t \in [\ell]$.
 \end{enumerate}
\end{theorem}

Before diving into the proof, let us first discuss some implications of the theorem.

First of all, Theorem \ref{thm:long-sequence-stable} implies that the upper bound in Theorem \ref{thm:upper-bound} is tight up to a factor of $O_k(\log n)$.
This follows from the simple observation that, if $\chi_t$ is $k$-stable and $\chi_{t-1} \succ \chi_t$, then $\refWL{k}{\chi_{t-1}} \equiv \chi_{t-1} \succeq \chi_t$, i.e., the sequence of colorings constructed in Theorem \ref{thm:long-sequence-stable} satisfies the requirements of Theorem \ref{thm:upper-bound}.

On the other hand, since all colorings $\chi_t$ are already $k$-stable, the theorem does not provide any lower bounds on the iteration number of $k$-WL.
However, Theorem \ref{thm:long-sequence-stable} still provides some valuable insights in this setting.
Indeed, all existing methods to bound the iteration number of $k$-WL \cite{KieferS19,LichterPS19} rely on ``parallelization arguments'', i.e., it is argued that at some point in the refinement process many color classes have to be split at the same time.
Theorem \ref{thm:long-sequence-stable} essentially implies that such arguments do not suffice to push the upper bounds on the iteration number beyond $O(n^{k-1})$ since such ``parallelization arguments'' typically also work in the extended setting of Theorem \ref{thm:upper-bound}.
As a concrete example, Kiefer and Schweitzer \cite{KieferS19} prove upper bounds on iteration number of $2$-WL by bounding the cost of a certain game related to $2$-WL.
This game naturally generalizes to $k$-WL, but Theorem \ref{thm:long-sequence-stable} immediately implies that its cost is $\Omega(n^{k-1})$ and thus, it is not possible to obtain improved upper bounds by analyzing said game.
So overall, Theorem \ref{thm:long-sequence-stable} can be interpreted as saying that, in order to obtain improved upper bounds on the iteration number of $k$-WL, we need to rely on arguments that also exploit the possibility of stabilization at an early point, and it is not possible to solely rely on ``parallelization arguments''.

\medskip

Let us now turn to the proof of Theorem \ref{thm:long-sequence-stable}.
It relies on the following theorem which provides a large set family with restricted intersections between its members.
Let $U$ be a set of size $n$.
A \emph{$k$-uniform set family (over $U$)} is a collection $\CF$ of $k$-element subsets of $U$. 

\begin{theorem}[{\cite[Theorem 4.11]{BabaiF20}}]
 \label{thm:set-family}
 For every $n \geq 2k^{2}$ there exists a $k$-uniform set family $\CF$ over a universe $U$ of $n$ points such that
 \begin{enumerate}
  \item $|E_1 \cap E_2| \leq k - 2$ for all distinct $E_1,E_2 \in \CF$, and
  \item $|\CF| \geq \left(\frac{n}{2k}\right)^{k-1}$.
 \end{enumerate}
\end{theorem}

Now, let $U$ be a universe of size $n \geq k$ and let $\CF$ be a $k$-uniform set family over $U$.
We set
\[V \coloneqq U \times \{0,1\}\]
and define a coloring $\chi_\CF\colon V^{k} \rightarrow C$ as follows.
Since the actual names of the colors are not relevant for our purposes, we only define the color classes, i.e., we specify when two tuples receive the same color.

Let $((u_1,a_1),\dots,(u_k,a_k)), ((u_1',a_1'),\dots,(u_k',a_k')) \in V^{k}$.
We define $\chi_\CF$ in such a way that $\chi_\CF((u_1,a_1),\dots,(u_k,a_k)) = \chi_\CF((u_1',a_1'),\dots,(u_k',a_k'))$ if and only if
\begin{enumerate}[label = (\Alph*)]
 \item\label{item:coloring-classes} $u_i = u_i'$ for all $i \in [k]$,
 \item\label{item:coloring-eq} $(u_i,a_i) = (u_j,a_j) \;\;\;\Leftrightarrow\;\;\; (u_i',a_i') = (u_j',a_j')$ for all $i,j \in [k]$, and
 \item\label{item:coloring-cfi} if $\{u_1,\dots,u_k\} \in \CF$, then $\sum_{i \in [k]} a_i \equiv \sum_{i \in [k]} a_i' \bmod 2$.
\end{enumerate}

\begin{lemma}
 \label{la:wl-stable}
 Suppose $|E_1 \cap E_2| \leq k - 2$ for all distinct $E_1,E_2 \in \CF$.
 Then $\chi_\CF$ is $k$-stable.
\end{lemma}

\begin{proof}
 Let $\vec v = ((u_1,a_1),\dots,(u_k,a_k)), \vec v' = ((u_1',a_1'),\dots,(u_k',a_k')) \in V^{k}$ such that
 \[\chi_\CF(\vec v) = \chi_\CF(\vec v').\]
 Observe that $u_i = u_i'$ for all $i \in [k]$ by Condition \ref{item:coloring-classes}.
 We need to show that the two tuples do not receive distinct colors after a single refinement step of $k$-WL, that is, we need to argue that
 \begin{align*}
        &\Big\{\!\!\Big\{ \big(\chi_\CF(\vec v[(u,a)/1]),\dots,\chi_\CF(\vec v[(u,a)/k])\big) \;\Big|\; u \in U, a \in \{0,1\} \Big\}\!\!\Big\}\\
  =\;\; &\Big\{\!\!\Big\{ \big(\chi_\CF(\vec v'[(u,a)/1]),\dots,\chi_\CF(\vec v'[(u,a)/k])\big) \;\Big|\; u \in U, a \in \{0,1\} \Big\}\!\!\Big\}
 \end{align*}
 where $\vec v[(u,a)/i]) = ((u_1,a_1),\dots,(u_{i-1},a_{i_1}),(u,a),(u_{i+1},a_{i+1}),\dots,(u_k,a_k))$ is the tuple obtained from $\vec v$ by replacing the $i$-th entry by $(u,a)$.
 Towards this end, we actually show the stronger statement that
 \begin{align*}
        &\Big\{\!\!\Big\{ \big(\chi_\CF(\vec v[(u,a)/1]),\dots,\chi_\CF(\vec v[(u,a)/k])\big) \;\Big|\; a \in \{0,1\} \Big\}\!\!\Big\}\\
  =\;\; &\Big\{\!\!\Big\{ \big(\chi_\CF(\vec v'[(u,a)/1]),\dots,\chi_\CF(\vec v'[(u,a)/k])\big) \;\Big|\; a \in \{0,1\} \Big\}\!\!\Big\}
 \end{align*}
 holds for all $u \in U$.
 
 Fix some $u \in U$.
 To see that these two multisets are equal, consider the set
 \[V' \coloneqq \{u_1,\dots,u_k,u\} \times \{0,1\} \subseteq V\]
 and the restriction $\chi_\CF' \colon (V')^{k} \rightarrow C\colon \vec v \mapsto \chi_{\CF}(\vec v)$ of $\chi_\CF$ to the set $(V')^{k}$.
 Also, let $\CF' \coloneqq \{E \in \CF \mid E \subseteq V'\}$.
 Since $|E_1 \cap E_2| \leq k - 2$ for all distinct $E_1,E_2 \in \CF$ and $|V'| \leq k+1$, we conclude that $|\CF'| \leq 1$.
 
 \begin{claim}
  There is a bijection $\varphi\colon V' \rightarrow V'$ such that
  \begin{enumerate}[label = (\roman*)]
   \item\label{item:bijection-color} $\chi_{\CF}'(\vec v) = \chi_{\CF}'(\varphi(\vec v))$ for all $\vec v \in (V')^{k}$, and
   \item\label{item:bijection-swap} $\varphi(u_i,a_i) = (u_i',a_i')$ for all $i \in [k]$.
  \end{enumerate}
 \end{claim}
 \begin{claimproof}
  For $i \in [k]$ we define $\varphi(u_i,a_i) \coloneqq (u_i,a_i')$ and $\varphi(u_i,1-a_i) \coloneqq (u_i,1-a_i')$.
  In particular, Condition \ref{item:bijection-swap} is satisfied since $u_i = u_i'$ for all $i \in [k]$.
  If there is some $E' \in \CF'$ such that $u \in E'$, then we define
  \[\varphi(u,a) \coloneqq \begin{cases}
                            (u,a)   &\text{if } \sum_{u_i \in E'} a_i \equiv \sum_{u_i \in E'} a_i' \bmod 2\\
                            (u,1-a) &\text{otherwise}
                           \end{cases}\]
  for both $a \in \{0,1\}$.
  If no such set $E' \in \CF'$ exists, then we set $\varphi(u,a) \coloneqq (u,a)$ for both $i \in \{0,1\}$.
  It can be easily verified that $\chi_{\CF}'(\vec v) = \chi_{\CF}'(\varphi(\vec v))$ for all $\vec v \in (V')^{k}$.
 \end{claimproof}

 Since the multisets above are defined in an isomorphism-invariant manner over the structure induced by $(V',\chi_F')$, we conclude that they have to be equal.
\end{proof}

\begin{proof}[Proof of Theorem \ref{thm:long-sequence-stable}]
 Let $\CF$ be the set family obtained from Theorem \ref{thm:set-family} and suppose that $\CF = \{E_1,\dots,E_\ell\}$.
 Observe that $\ell \geq \left(\frac{n}{2k}\right)^{k-1}$ as desired.
 For $t \in [0,\ell]$ we define $\CF_t \coloneqq \{E_1,\dots,E_t\} \subseteq \CF$ and $\chi_t \coloneqq \chi_{\CF_t}$.
 Then $\chi_t$ is $k$-stable by Lemma \ref{la:wl-stable}.
 Also, $\CF_{t-1} \subset \CF_t$ which implies that $\chi_{t-1} \succ \chi_t$ by definition of the coloring $\chi_t$.
 Finally, it is easy to verify that all colorings are shufflable and compatible with equality. 
\end{proof}

%% file: lower-bounds.tex
In this section, we obtain improved lower bounds on the iteration number of the Weisfeiler-Leman algorithm.
More precisely, we prove Theorem \ref{thm:wl-round-lower-bound}.
Our proof strategy is similar to the one employed by Berkholz and Nordstr{\"{o}}m in \cite{BerkholzN16}.
First, for every sufficiently large $\ell_\hi \geq \ell_\lo$, we construct pairs of structures that can be distinguished by $\ell_\lo$-WL, but $\ell_\hi$-WL still requires a linear number of iterations to distinguish them.
Afterwards, we apply a \emph{hardness compression} that reduces the number of vertices in the obtained structures while preserving the iteration number of the Weisfeiler-Leman algorithm.
Actually, for the second step, we can rely on the same tools that are already used by Berkholz and Nordstr{\"{o}}m in \cite{BerkholzN16}.

\subsection{Overview}

The hard instances we construct are based on propositional XOR-formulas that can also be viewed as systems of linear equations over the $2$-element field $\FF_2$.

Let $V$ be a finite set which we interpret as a set of variables that take values in $\{0,1\}$.
An \emph{XOR-constraint (over $V$)} is a pair $(C,a)$ where $C \subseteq V$ and $a \in \{0,1\}$.
The reader is encouraged to think of such a constraint as the equation $x_1 + \dots + x_k \equiv a \bmod 2$ where $C = \{x_1,\dots,x_k\}$ is the set of those variables that appear on the left side of the equation.
We explicitly allow $C$ to be empty; $(\emptyset,0)$ is always satisfied and $(\emptyset,1)$ is unsatisfiable.
Let $\CC$ be a set of XOR-constraints.
We define the \emph{arity} of $\CC$ to be the maximum cardinality of $C$ for any pair $(C,a) \in \CC$.

We can translate a set of XOR-constraints into a pair of relational structures as follows.
Let $\CC$ be a set of XOR-constraints over a set $V$.
Also suppose that $V = \{x_1,\dots,x_n\}$.
We define $\FA = \FA(\CC)$ and $\FB = \FB(\CC)$ as follows.
We set $V(\FA) = V(\FB) \coloneqq V \times \{0,1\}$, i.e., each element of the structures $\FA$ and $\FB$ corresponds to an assignment of a single variable.
For each $i \in [n]$, we add a unary relation $X_i$ and set $X_i^{\FA} = X_i^{\FB} \coloneqq \{(x_i,0),(x_i,1)\}$.
Finally, for every constraint $(C,a) \in \CC$ with $C = \{x_{i_1},\dots,x_{i_k}\}$ we introduce a $k$-ary relation $R_{C,a}$ and define
\[R_{C,a}^{\FA} \coloneqq \left\{\Big((x_{i_1},b_1),\dots,(x_{i_k},b_k)\Big) \;\middle|\; b_1,\dots,b_k \in \{0,1\}, \sum_{j = 1}^{k} b_j \equiv 0 \mod 2\right\}\]
and
\[R_{C,a}^{\FB} \coloneqq \left\{\Big((x_{i_1},b_1),\dots,(x_{i_k},b_k)\Big) \;\middle|\; b_1,\dots,b_k \in \{0,1\}, \sum_{j = 1}^{k} b_j \equiv a \mod 2\right\}\]

Instead of analysing the Weisfeiler-Leman algorithm directly on $\FA(\CC)$ and $\FB(\CC)$, it turns out to more convenient to consider the following game that is directly played on $\CC$ and is known to capture the same information as applying the Weisfeiler-Leman algorithm to the associated structures.

Let $\CC$ be a set of XOR-constraints over a set $V$.
Let $k \in \mathbb{N}$ such that $\CC$ has arity at most $k$.
A partial assignment $\beta \colon X \rightarrow \{0,1\}$ with $X \subseteq V$ \emph{violates} an XOR-constraint $(C,a) \in \CC$ if $C \subseteq X$ and
\begin{equation}
 \sum_{x \in C} \beta(x) \not\equiv a \mod 2.
\end{equation}
For a partial assignment $\beta_0 \colon X_0 \rightarrow \{0,1\}$ with $|X_0| \leq k$ the \emph{$r$-round $k$-pebble game} $\CG_k^{r}(V,\CC,\beta_0)$ is played as follows:
\begin{itemize}
 \item The game has two players called Verifier and Falsifier.
 \item The game is played in rounds with initial position $\beta_0$.
 \item Suppose $\beta\colon X \rightarrow \{0,1\}$ is the current position. Then the next round consists of the following steps:
  \begin{itemize}
   \item Falsifier chooses $x \in V \setminus X$ and $X' \subseteq X$ such that $|X' \cup \{x\}| \leq k$.
   \item Verifier chooses $b \in \{0,1\}$.
   \item The game moves to position $\beta'\colon X' \cup \{x\} \rightarrow \{0,1\}$ with $\beta'(x') = \beta(x')$ for $x' \in X'$ and $\beta'(x) = b$.
  \end{itemize}
 \item Falsifier wins a play if within the first $r$ rounds an assignment $\beta$ violates some XOR-constraint $(C,a) \in \CC$ (if $r = 0$, then Falsifier wins if the initial assignment $\beta_0$ violates some constraint in $\CC$).
 \item Verifier wins a play if Falsifier does not win within the first $r$ rounds.
\end{itemize}

We say Falsifier (respectively Verifier) wins the game $\CG_k^{r}(V,\CC,\beta_0)$ if Falsifier (respectively Verifier) has a winning strategy for the game.
The \emph{$k$-pebble game} $\CG_k(V,\CC,\beta_0)$ is played in the same way, but without any restriction on the number of rounds played.

The following lemma relates the pebble game $\CG_{k}^{r}(V,\CC,\emptyset)$ to bounded-variable fragments of first-order logic and thereby, using Corollary \ref{cor:wl-logic-distinguish-structures}, also to the Weisfeiler-Leman algorithm.
(Here, we use $\emptyset$ to denote the empty assignment, i.e., the domain $X_0$ of the initial partial assignment $\beta_0$ is empty.)

\begin{lemma}[{\cite[Lemma 2.1]{BerkholzN16}}]
 \label{la:translate-bounds}
 Let $k,r \in \NN$ such that $r > 0$ and $k \geq 3$.
 Let $\CC$ be a set of XOR-constraints over a universe $V$ of arity at most $k$.
 Then the following statements are equivalent:
 \begin{enumerate}[label = (\roman*)]
  \item Falsifier wins the $r$-round $k$-pebble game $\CG_{k}^{r}(V,\CC,\emptyset)$.
  \item There exists a sentence $\varphi \in \LLkq{k}{r}$ such that $\varphi \models \FA(\CC)$ and $\varphi \not\models \FB(\CC)$.
  \item There exists a sentence $\varphi \in \LCkq{k}{r}$ such that $\varphi \models \FA(\CC)$ and $\varphi \not\models \FB(\CC)$.
 \end{enumerate}
\end{lemma}

To obtain a set of XOR-constraints on which Falsifier requires a large number of rounds to win the pebble game, we proceed in two steps.
First, for every sufficiently large $\ell_\hi \geq \ell_\lo$, we construct a set of XOR-constraints such that Falsifier wins the $\ell_\lo$-pebble game, but still requires a linear number of rounds to win the $\ell_\hi$-pebble game.
This is formalized by the next lemma which forms the main technical contribution of this section.

\begin{lemma}
 \label{la:many-rounds}
 There are absolute constants $\ell_\lo \geq 2$ and $\delta > 1$ such that for every $\ell_\hi \geq \ell_\lo$ and every $r \geq 1$
 there is a set of XOR-constraints $\CC$ of arity at most $\ell_\lo$ over a set $V$ of size $|V| \leq \delta \cdot \ell_\hi^{2} \cdot r$ such that Falsifier
 \begin{enumerate}[label = (\alph*)]
  \item wins the $\ell_\lo$-pebble game $\CG_{\ell_\lo}(V,\CC,\emptyset)$, but
  \item does not win the $r$-round $\ell_\hi$-pebble game $\CG_{\ell_\hi}^{r}(V,\CC,\emptyset)$.
 \end{enumerate}
\end{lemma}

We remark that a similar result has also been obtained in \cite{BerkholzN16}, but with weaker guarantees on the number of rounds required to win the $\ell_\hi$-pebble game.
It is exactly this improvement that allows us to obtain stronger lower bounds on the iteration number of $k$-WL in comparison to \cite{BerkholzN16}.

\begin{remark}
 \label{rem:dummy-variables}
 When applying Lemma \ref{la:many-rounds}, we also require that $|V| \geq r$ which is not guaranteed by the lemma.
 However, if $|V| < r$ we can simply add dummy variables that do not appear in any constraint to increase the number of variables.
 It is easy to see that all properties guaranteed by the lemma remain valid.
 In particular, the dummy variables do not affect the winning strategy for either player (if Falsifier asks for an assignment of a dummy variable, Verifier simply chooses any value; since dummy variables do not appear in any constraints this is always safe).
\end{remark}

Afterwards, we rely on the following hardness compression lemma that reduces the number of variables while essentially maintaining the number of rounds that Falsifier requires to win the game.

\begin{lemma}[Berkholz, Nordstr{\"{o}}m {\cite[Lemma 3.3]{BerkholzN16}}]
 \label{la:condensation}
 There is an absolute constant $\Delta_0 \geq 1$ such that the following holds.
 Suppose $\CC$ is a set of XOR-constraints of arity at most $p$ over a set $V$ of size $|V| = m$.
 Also assume there are parameters $\ell_\lo > 0$, $\ell_\hi \geq \Delta_0 \ell_\lo$ and $r > 0$ such that Falsifier
 \begin{enumerate}[label = (\alph*)]
  \item wins the $\ell_\lo$-pebble game $\CG_{\ell_\lo}(V,\CC,\emptyset)$, but
  \item does not win the $r$-round $\ell_\hi$-pebble game $\CG_{\ell_\hi}^{r}(V,\CC,\emptyset)$.
 \end{enumerate}
 Let $\Delta$ be an integer such that $\Delta_0 \leq \Delta \leq \ell_\hi/\ell_\lo$ and $(2\ell_\hi\Delta)^{2\Delta} \leq m$.
 Then there is a set of XOR-constraints $\CD$ of arity at most $\Delta p$ over a set $W$ of size $|W| = \lceil m^{3/\Delta}\rceil$ such that Falsifier
 \begin{enumerate}[label = (\Alph*)]
  \item wins the $(\Delta \ell_\lo)$-pebble game $\CG_{\Delta \ell_\lo}(W,\CD,\emptyset)$, but
  \item does not win the $\frac{r}{2\ell_\hi}$-round $\ell_\hi$-pebble game $\CG_{\ell_\hi}^{r/(2\ell_\hi)}(W,\CD,\emptyset)$.
 \end{enumerate}
\end{lemma}

Combining Lemmas \ref{la:many-rounds} and \ref{la:condensation}, we obtain the following corollary.

\begin{corollary}
 \label{cor:maximal-rounds}
 There are absolute constants $k_0 \in \NN$ and $\alpha,\varepsilon > 0$ such that for every $d \geq k \geq k_0$ and every $n \geq \alpha \cdot d^8 \cdot k^6$
 there is a set of XOR-constraints $\CC$ of arity at most $k$ over a set $V$ of size $|V| \leq n$ such that Falsifier wins the $k$-pebble game $\CG_{k}(V,\CC,\emptyset)$, but does not win the $r$-round $d$-pebble game $\CG_{d}^{r}(V,\CC,\emptyset)$ for all $r \leq n^{\varepsilon k}$.
\end{corollary}

\begin{proof}
 Let $\ell_\lo \geq 2$ and $\delta > 1$ denote the constants from Lemma \ref{la:many-rounds}.
 Also, let $\Delta_0$ denote the constant from Lemma \ref{la:condensation} and suppose without loss of generality that $\delta, \Delta_0$ are integers and $\Delta_0 \geq 3$.
 We choose
 \[k_0 \coloneqq \max\{\Delta_0\ell_\lo, 6\ell_\lo\}.\]
 
 Let $d \geq k \geq k_0$.
 We set $p \coloneqq \ell_\lo$, $\ell_\hi \coloneqq d$ and $\Delta \coloneqq \lfloor\frac{k}{\ell_\lo}\rfloor$.
 We have $\ell_\hi = d \geq k \geq k_0 \geq \Delta_0\ell_\lo$ and $\Delta_0 \leq \frac{k}{\ell_\lo}$.
 Since $\Delta_0$ is an integer, we conclude that $\Delta_0 \leq \Delta$.
 
 We define
 \[n_0 \coloneqq \max\{\left(\delta \cdot \ell_\hi^{2} \cdot (2\ell_\hi\Delta)^{2\Delta}\right)^{3/\Delta},4 \cdot \delta \cdot \ell_\hi^{3}\}\]
 and set $\alpha \coloneqq \max\{64\cdot \delta \cdot \ell_\lo^{-6},4\cdot\delta\}$.
 Let $n \geq \alpha \cdot d^8 \cdot k^6$.
 Using $\Delta \geq \Delta_0 \geq 3$, we get that
 \[\left(\delta \cdot \ell_\hi^{2} \cdot (2\ell_\hi\Delta)^{2\Delta}\right)^{3/\Delta} \leq \delta \cdot \ell_\hi^{2} \cdot (2\ell_\hi\Delta)^6 \leq 64 \cdot \delta \cdot \ell_\hi^8 \cdot \left(\frac{k}{\ell_\lo}\right)^6 \leq \alpha \cdot d^8 \cdot k^6\]
 and $4 \cdot \delta \cdot \ell_\hi^{3} \leq \alpha \cdot d^3$.
 So in particular $n \geq n_0$.
 Let $r$ be the maximal integer such that
 \begin{equation}
  \label{eq:choose-r}
  \left(\delta \cdot \ell_\hi^{2} \cdot r\right)^{3/\Delta} \leq n.
 \end{equation}
 Note that $r \geq (2\ell_\hi\Delta)^{2\Delta}$ since $n \geq n_0$.
 Let $\CC$ be the set of XOR-constraints of arity at most $\ell_\lo$ over a set $V$ of size $|V| \leq \delta \cdot \ell_\hi^{2} \cdot r$ obtained from Lemma \ref{la:many-rounds}.
 By adding dummy variables (see Remark \ref{rem:dummy-variables}), we may assume without loss of generality that $m \coloneqq |V| \geq r \geq (2\ell_\hi\Delta)^{2\Delta}$.
 
 By applying Lemma \ref{la:condensation}, we obtain a set of XOR-constraints $\CD$ of arity at most $\Delta p$ over a set $W$ of size $|W| = \lceil m^{3/\Delta}\rceil$ such that Falsifier
 \begin{enumerate}[label = (\Alph*)]
  \item wins the $(\Delta \ell_\lo)$-pebble game $\CG_{\Delta \ell_\lo}(W,\CD,\emptyset)$, but
  \item does not win the $\frac{r}{2\ell_\hi}$-round $\ell_\hi$-pebble game $\CG_{\ell_\hi}^{r/(2\ell_\hi)}(W,\CD,\emptyset)$.
 \end{enumerate}
 First observe that $\Delta p = \lfloor\frac{k}{\ell_\lo}\rfloor \cdot \ell_\lo \leq k$ and
 \[|W| = \left\lceil m^{3/\Delta}\right\rceil \leq \left\lceil\left(\delta \cdot \ell_\hi^{2} \cdot r\right)^{3/\Delta}\right\rceil \leq n.\]
 Since $\Delta \ell_\lo \leq k$, it holds that Falsifier wins the $k$-pebble game $\CG_{k}(W,\CD,\emptyset)$.
 Moreover, Falsifier does not win the $\frac{r}{2\ell_\hi}$-round $d$-pebble game $\CG_{d}^{r/(2\ell_\hi)}(W,\CD,\emptyset)$.
 We have that
 \[\left(\delta \cdot \ell_\hi^{2} \cdot 2r\right)^{3/\Delta} \geq n\]
 since $r$ is the maximal integer to satisfy Equation \eqref{eq:choose-r}.
 This implies that
 \[\frac{r}{2\ell_\hi} \geq \frac{n^{\Delta/3}}{4 \cdot \delta \cdot \ell_\hi^{3}} \geq n^{\frac{1}{3}\lfloor\frac{k}{\ell_\lo}\rfloor - 1} \geq n^{\varepsilon k}\]
 for some sufficiently small absolute constant $\varepsilon > 0$.
\end{proof}

With Corollary \ref{cor:maximal-rounds} at hand, we are now ready to prove Theorems \ref{thm:wl-round-lower-bound} and \ref{thm:quantifier-rank-lower-bound}.

\begin{proof}[Proof of Theorem \ref{thm:wl-round-lower-bound}]
 Let $k_0' \in \NN$ and $\alpha',\varepsilon' > 0$ denote the absolute constants from Corollary \ref{cor:maximal-rounds}.
 
 Let $k_0 \coloneqq \max\{k_0',3\}$.
 We set $\alpha \geq 1$ and $\varepsilon > 0$ in such a way that for all $d \geq k \geq k_0$ and $n \geq \alpha d^8 k^6$ it holds that
 \[\left\lfloor\frac{n}{2}\right\rfloor \geq \alpha'(d+1)^8(k+1)^6\]
 and
 \[\left(\frac{n}{2} - 1\right)^{\varepsilon' k} - d \geq n^{\varepsilon k}.\]
 Now, let us fix some $d \geq k \geq k_0$ and $n \geq \alpha d^8 k^6$.
 Let $d' \coloneqq d+1$, $k' \coloneqq k$ and $n' \coloneqq \lfloor\frac{n}{2}\rfloor$.
 We apply Corollary \ref{cor:maximal-rounds} with parameters $d',k',n'$ and obtain a set of XOR-constraints $\CC$ of arity at most $k'$ over a set $V'$ of size $|V'| \leq n'$ such that Falsifier wins the $k'$-pebble game $\CG_{k'}(V',\CC,\emptyset)$, but does not win the $r'$-round $d'$-pebble game $\CG_{d'}^{r'}(V',\CC,\emptyset)$ for $r' = (n')^{\varepsilon' k'}$.
 
 Let $\FA \coloneqq \FA(\CC)$ and $\FB \coloneqq \FB(\CC)$.
 Then $|V(\FA)| = |V(\FB)| = 2|V'| \leq 2n' \leq n$.
 Note that we can easily increase the size of both structures by adding isolated elements that do not participate in any relations.
 Also, note that both structures have arity at most $k' = k$.
 
 By Lemma \ref{la:translate-bounds} and Corollary \ref{cor:wl-logic-distinguish-structures}, $k$-WL distinguishes between $\FA$ and $\FB$.
 On the other hand, again by Lemma \ref{la:translate-bounds} and Corollary \ref{cor:wl-logic-distinguish-structures}, $d$-WL does not distinguish $\FA$ and $\FB$ after $r \coloneqq r' - d$ refinement rounds.
 We get that
 \[r = r' - d = (n')^{\varepsilon' k'} - d \geq \left(\frac{n}{2} - 1\right)^{\varepsilon' k} - d \geq n^{\varepsilon k}.\qedhere\]
\end{proof}

\begin{proof}[Proof of Theorem \ref{thm:quantifier-rank-lower-bound}]
 This follows directly from Corollary \ref{cor:maximal-rounds} and Lemma \ref{la:translate-bounds}.
\end{proof}

The remainder of this section is devoted to the proof of Lemma \ref{la:many-rounds}.

\subsection{The Closure of the Constraint Set}
\label{sec:closure-operator}

The critical step in the proof of Lemma \ref{la:many-rounds} is to argue that Verifier survives a linear number of rounds even for a large number of pebbles.
Here, we rely on an alternative description of winning positions in terms of a closure operator.

Let $k \in \mathbb{N}$.
Let $V$ be a finite set and let $\CC$ be a set of XOR-constraints over $V$ of arity at most $k$.
We define the \emph{$k$-attractor}
\[\attr_k(\CC) \coloneqq \CC \cup \Big\{(C_1 \oplus C_2,a_1 + a_2 \bmod 2) \;\Big|\; (C_1,a_1),(C_2,a_2) \in \CC, |C_1 \oplus C_2| \leq k\Big\}.\]
Here, $C_1 \oplus C_2$ denotes the symmetric difference between the two sets, that is, $C_1 \oplus C_2 \coloneqq (C_1 \cup C_2) \setminus (C_1 \cap C_2)$.

Intuitively speaking, if $x_1 + \dots + x_\ell \equiv a_1 \bmod 2$ and $y_1 + \dots + y_m \equiv a_2 \bmod 2$ are two constraints in $\CC$,
then every satisfying assignment also has to satisfy the equation $x_1 + \dots + x_\ell + y_1 + \dots + y_m \equiv a_1 + a_2 \bmod 2$.
Since all variables appearing in both sets $\{x_1,\dots,x_k\}$ and $\{y_1,\dots,y_m\}$ cancel over $\FF_2$, we only need to keep those variables appearing in the symmetric difference.
In the case that the resulting number of variables is bounded by $k$, we add the corresponding equation to the $k$-attractor of the constraint set.

We define $\cl_k^{(0)}(\CC) \coloneqq \CC$ and $\cl_k^{(r+1)}(\CC) \coloneqq \attr_k(\cl_k^{(r)}(\CC))$ for all $r \geq 0$.
Finally, we define the \emph{$k$-closure} of $\CC$ to be the set $\cl_k(\CC) \coloneqq \cl_k^{(r)}(\CC)$ for the minimal $r \geq 0$ such that $\cl_k^{(r+1)}(\CC) = \cl_k^{(r)}(\CC)$.

The following lemma provides the key method to prove that Verifier can survive a certain number of rounds.

\begin{lemma}
 \label{la:dupl-closure-to-consistency}
 Let $\beta\colon X \rightarrow \{0,1\}$ be a partial assignment with $|X| \leq k$ such that $\beta$ violates no XOR-constraint $(C,a) \in \cl_k^{(r)}(\CC)$.
 Then Verifier wins $\CG_k^{r}(V,\CC,\beta)$.
\end{lemma}

\begin{proof}
 We prove the statement by induction on $r$.
 For $r = 0$ the statement is trivial.
 So suppose $r \geq 1$ and Falsifier chooses $x \in V \setminus X$ and $X' \subseteq X$ such that $|X' \cup \{x\}| \leq k$ in the first round.
 For $b \in \{0,1\}$ let $\beta_b\colon X' \cup \{x\} \rightarrow \{0,1\}$ be the partial assignment with $\beta_b(x') = \beta(x')$ for $x' \in X'$ and $\beta_b(x) = b$.
 Assume towards a contradiction that, for every $b \in \{0,1\}$, there is some XOR-constraint $(C_b,a_b) \in \cl_k^{(r-1)}(\CC)$ violated by $\beta_b$.
 Observe that $x \in C_b$ for both $b \in \{0,1\}$ (since otherwise $\beta$ would violate $(C_b,a_b)$ contradicting our assumption).
 Let $C \coloneqq C_0 \oplus C_1 \subseteq X$ and $a \coloneqq (a_0 + a_1) \bmod 2$.
 Note that $|C| \leq k$ since $C \subseteq X$. 
 Then \[\sum_{y \in C} \beta(y) \equiv \sum_{y \in C_0\setminus\{x\}}\beta(y) + \sum_{y \in C_1\setminus\{x\}} \beta(y) \equiv 1 + \sum_{y \in C_0}\beta_0(y) + \sum_{y \in C_1} \beta_1(y) \equiv 1 + a \mod 2\]
 and $(C,a) \in \cl_k^{(r)}(\CC)$.
 Hence, $\beta$ violates some $(C,a) \in \cl_k^{(r)}(\CC)$ which is a contradiction.
 
 So there is some $b \in \{0,1\}$ such that $\beta_b$ violates no XOR-constraint in $\cl_k^{(r-1)}(\CC)$.
 Verifier chooses such a $b \in \{0,1\}$ and the game moves to position $\beta_b$ which violates no XOR-constraint in $\cl_k^{(r-1)}(\CC)$.
 So Verifier wins $\CG_k^{r-1}(V,\CC,\beta_b)$ by the induction hypothesis which implies that Verifier also wins $\CG_k^{r}(V,\CC,\beta)$.
\end{proof}

\subsection{Layered Graphs and Expansion}

Next, we discuss the construction of certain expander graphs.
Overall, we are aiming to construct what we refer to as \emph{single-neighbor layered expanders}.
Towards this end, we start with constructing standard bipartite expander graphs with an expansion that is close to the minimum degree of one side of the bipartite graph.
We then define \emph{single-neighbor expanders} and observe that bipartite expanders with large expansion also are single-neighbor expanders (with a slightly smaller expansion parameter).
Finally, we obtain \emph{single-neighbor layered expanders} by ``stacking single-neighbor expanders on top of each other''.

\subsubsection{Expander Graphs}

We start by defining standard bipartite expander graphs.

\begin{definition}
 Let $0 < \gamma < 1$ and $\alpha > 1$ be constants and let $G = (V,W,E)$ be a bipartite graph.
 We say that $G$ is an \emph{$(\alpha,\gamma)$-expander} if for every $\emptyset \neq Y \subseteq W$ with $|Y| \leq \gamma |W|$ it holds that \[N(Y) \geq \alpha |Y|.\]
\end{definition}

For more information on expander graphs we refer to \cite{MotwaniR95,Vadhan12}.
The references also contain variants of the following standard argument that guarantees the existence of graphs with good expansion properties.
For our purposes, the crucial property in the lemma below is that the expansion $\alpha$ is relatively close to the degree of the vertices in $W$.

\begin{lemma}
 \label{la:expander}
 There is some number $R_0 \geq 2$ such that for every $r \geq R_0$ and every $n \geq 4r$ there is a $(\frac{3}{4}r,\frac{1}{20r})$-expander $G = (V,W,E)$ such that $|V| = |W| = n$ and $\deg(w) = r$ for all $w \in W$.
\end{lemma}

\begin{proof}
 Suppose $r$ is sufficiently large.
 Let $V,W$ be two sets with $|V| = |W| \geq 4r$.
 We construct a bipartite graph $G = (V,W,E)$ using the following random process:
 for each $w \in W$ we select independently and uniformly at random a set of $r$ distinct neighbors from $V$.
 We prove that, for $r$ sufficiently large, with positive probability the graph $G$ is a $(\frac{3}{4}r,\frac{1}{20r})$-expander.

 Let $n \coloneqq |V| = |W|$. For $X \subseteq V$ and $Y \subseteq W$ let~$p_{X,Y}$ denote the probability that $N(Y) \subseteq X$.
 Then
 \[p_{X,Y} \leq \left(\frac{|X|}{n}\right)^{r\cdot |Y|}.\]
 Furthermore, let $\alpha \coloneqq \frac{3}{4}r$ and $\gamma \coloneqq \frac{1}{20 r}$.
 Let $p$ be the probability that $G$ is not a $(\gamma,\alpha)$-expander.
 Then, using the inequality~$\binom{n}{k}\leq (ne/k)^k$, we get
 \begin{align*}
  p &\leq \sum_{\substack{Y \subseteq W\\|Y| \leq \gamma\cdot n}}\;\sum_{\substack{X \subseteq V\\|X| = \lfloor \alpha|Y|\rfloor}} p_{X,Y}\\
    &\leq \sum_{s = 1}^{\lfloor \gamma\cdot n\rfloor} \sum_{\substack{Y \subseteq W\\|Y| =s}}\;\sum_{\substack{X \subseteq V\\|X| =  \lfloor\alpha|Y|\rfloor}} \left(\frac{|X|}{n}\right)^{r\cdot |Y|}\\
    &\leq \sum_{s = 1}^{\lfloor \gamma\cdot n\rfloor} \binom{n}{s}\binom{n}{\lfloor \alpha s \rfloor} \left(\frac{\alpha s}{n}\right)^{r\cdot s}\\
    &\leq \sum_{s = 1}^{\lfloor \gamma\cdot n\rfloor} \left(\frac{ne}{s}\right)^{s} \left(\frac{ne}{\alpha s}\right)^{\alpha \cdot s} \left(\frac{\alpha s}{n}\right)^{r\cdot s}\\
    &= \sum_{s = 1}^{\lfloor \gamma\cdot n\rfloor} \left[\left(\frac{ne}{s}\right) \left(\frac{ne}{\alpha s}\right)^{\alpha} \left(\frac{\alpha s}{n}\right)^{r}\right]^{s}
    \\   
    &= \sum_{s = 1}^{\lfloor \gamma\cdot n\rfloor} \left[\left(\frac{s}{n}\right)^{r-\alpha-1} e^{1+\alpha} \alpha^{r-\alpha}\right]^{s} \\
    &= \sum_{s = 1}^{\lfloor \gamma\cdot n\rfloor} \left[\left(\frac{s}{n}\right)^{r/4-1} e^{1+3r/4} (3r/4)^{r/4}\right]^{s}\\
    &\leq \sum_{s = 1}^{\lfloor \gamma\cdot n\rfloor} \left[\gamma^{r/4-1} e^{1+3r/4} (3r/4)^{r/4}\right]^{s}\;.
 \end{align*}
 Now let $x := \gamma^{r/4-1} e^{1+3r/4} (3r/4)^{r/4}$.
 For~$r$ sufficiently large we get \[x = (20r)^{1-r/4} e^{1+3r/4} (3r/4)^{r/4} = 20er \left(\frac{3e^{3}}{80}\right)^{r/4} <1/10.\] 
 It follows that \[p \leq \sum_{s = 1}^{\infty} x^{s} = \frac{x}{1-x} \leq \frac{1}{9}.\]
 In particular, $p < 1$ which implies the existence of the desired expander graph.
\end{proof}

Next, we turn to what we call single-neighbor expanders where each sufficiently small set $Y \subseteq V$ is required to have a large number of neighbors that additionally satisfy the property that they are the neighbor of only a single vertex from $Y$.
Let $G = (V,W,E)$ be a bipartite graph.
For $Y \subseteq W$ we define
\[N^{*}(Y) = \{v \in N(Y) \mid |N(v) \cap Y| = 1\}.\]

\begin{definition}
 Let $0 < \gamma < 1$ and $\alpha > 1$ be constants and let $G = (V,W,E)$ be a bipartite graph.
 We say that $G$ is an \emph{$(\alpha,\gamma)$-single-neighbor expander} if for every $\emptyset \neq Y \subseteq W$ with $|Y| \leq \gamma |W|$ it holds that \[N^{*}(Y) \geq \alpha |Y|.\]
\end{definition}

We can obtain single-neighbor expanders from Lemma \ref{la:expander} by allowing some loss on the expansion parameter $\alpha$.

\begin{corollary}
 \label{cor:single-neighbor-expander}
 There is some number $R_0 \geq 5$ such that for every $r \geq R_0$ and every $n \geq 4r$ there is a $(\frac{1}{4}r,\frac{1}{20r})$-single-neighbor expander $G = (V,W,E)$ such that $|V| = |W| = n$ and $\deg(w) = r$ for all $w \in W$.
\end{corollary}

\begin{proof}
 Choose $R_0 \coloneqq \max(5,R_0')$ where $R_0'$ is the constant from Lemma \ref{la:expander} and suppose $r \geq R_0$ and $n \geq 4r$.
 By Lemma \ref{la:expander}, there is a $(\frac{3}{4}r,\frac{1}{20r})$-expander $G = (V,W,E)$ such that $|V| = |W| = n$ and $\deg(w) = r$ for all $w \in W$.
 We claim that $G$ is a $(\frac{1}{4}r,\frac{1}{20r})$-single-neighbor expander.
 Let $Y \subseteq W$ with $|Y| \leq \frac{n}{20r}$.
 Then $|N(Y)| \geq \frac{3}{4}r|Y|$.
 Furthermore $|N(Y)| = |N^{*}(Y)| + |\{v \in N(Y) \mid |N(v) \cap Y| \geq 2\}| \leq |N^{*}(Y)| + \frac{1}{2}r|Y|$ because every vertex in $Y$ has degree $r$.
 Thus, $|N^{*}(Y)| \geq \frac{1}{4}r|Y|$.
\end{proof}

\subsubsection{Layered Graphs}

Now, we turn to the construction of single-neighbor layered expanders which is the main tool for constructing the desired constraint sets in the proof of Lemma \ref{la:many-rounds}.
We start by defining a certain notion of layered graphs (see also Figure \ref{fig:layered-graph}).

Let $\ell,m \in \NN$.
An \emph{$(\ell \times m)$-layered graph} is a bipartite graph $G = (V,W,E)$ for which there are partitions $V = V_0 \uplus \dots \uplus V_\ell$ and $W = W_1 \uplus \dots \uplus W_\ell$ such that
\begin{enumerate}
 \item $|V_i| = m$ for all $i \in [0,\ell]$,
 \item $|W_i| = m$ for all $i \in [\ell]$,
 \item $N_G(W_i) \subseteq V_{i-1} \cup V_i$ for all $i \in [\ell]$, and
 \item $G[V_i \cup W_i]$ is $1$-regular (i.e., a matching) for all $i \in [\ell]$.
\end{enumerate}

\begin{figure}
 \centering
 \begin{tikzpicture}
  \draw[thick,gray] (0,0) ellipse (4cm and 0.3cm);
  \draw[thick,gray] (0,1) ellipse (4cm and 0.3cm);
  \draw[thick,gray] (0,2) ellipse (4cm and 0.3cm);
  \draw[thick,gray] (0,3) ellipse (4cm and 0.3cm);
  \draw[thick,gray] (0,5) ellipse (4cm and 0.3cm);
  \draw[thick,gray] (0,6) ellipse (4cm and 0.3cm);
  
  \node at (5,0) {$V_0$};
  \node at (5,1) {$W_1$};
  \node at (5,2) {$V_1$};
  \node at (5,3) {$W_2$};
  \node at (5,5) {$W_\ell$};
  \node at (5,6) {$V_\ell$};
  
  \node at (-2,4) {$\vdots$};
  \node at (2,4) {$\vdots$};
  
  \foreach \i in {0,1,2,3,5,6}{
   \foreach \j in {0,1,2,3,4,8,9,10}{
    \node[smallvertex] (v-\i-\j) at (-3 + \j*0.6,\i) {};
   }
   \node at (0.6,\i) {$\dots$};
  }
  
  \foreach \i/\j in {1/2,5/6}{
   \foreach \k in {0,1,2,3,4,8,9,10}{
    \draw (v-\i-\k) edge (v-\j-\k);
   }
  }
  
  \foreach \i/\j in {0/1,2/3}{
   \foreach \k/\m in {0/0,0/2,1/0,1/1,1/2,2/1,2/3,3/2,3/3,3/4,4/3,4/4,8/8,8/9,9/8,9/10,10/9,10/10}{
    \draw (v-\j-\k) edge (v-\i-\m);
   }
  }
  
 \end{tikzpicture}
 \caption{Visualization of $(\ell \times m)$-layered graphs.}
 \label{fig:layered-graph}
\end{figure}
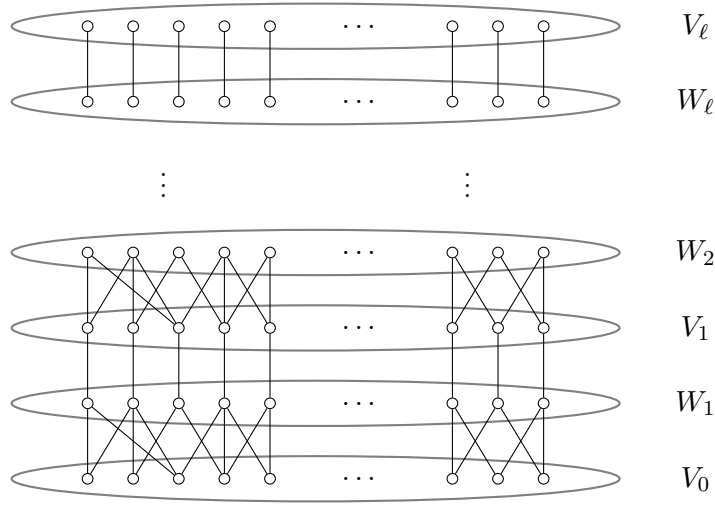

With this, we are now ready to define the notion of single-neighbor layered expanders.

\begin{definition}
 Let $\ell,m \geq 2$.
 Let $0 < \gamma < 1$ and $\alpha > 1$ be constants and let $G = (V,W,E)$ be an $(\ell \times m)$-layered graph.
 We say that $G$ is an \emph{$(\alpha,\gamma)$-single-neighbor $(\ell \times m)$-layered expander} if for every $\emptyset \neq Y \subseteq W$ with $|Y| \leq \gamma m$ it holds that \[N^{*}(Y) \geq \alpha |Y|.\]
\end{definition}

Note that an $(\alpha,\gamma)$-single-neighbor $(\ell \times m)$-layered expander is not a $(\alpha,\gamma)$-single-neighbor expander since we are only considering sets $Y \subseteq W$ of size $|Y| \leq \gamma m$, i.e., we are only considering sets that are smaller (by a factor of $\gamma$) than a single layer of a layered graph.
In particular, the reader is encouraged to think of $\ell$ being much larger than $m$.
In this case, such a graph is far from being a (global) expander, but the key property is that it behaves like an expander when only considering a few layers of the graph.

By again allowing some small loss on the expansion parameter $\alpha$, we can obtain single-neighbor layered expanders by ``stacking $\ell$ copies of a single-neighbor expander on top of each other''.

\begin{corollary}
 \label{cor:single-neighbor-layered-expander}
 There is some number $R_0 \geq 9$ such that for every $r \geq R_0$, every $\ell \geq 1$, and every $m \geq 4r$ there is a $(\frac{1}{4}r - 1,\frac{1}{20r})$-single-neighbor $(\ell \times m)$-layered expander $G = (V,W,E)$ with partitions $V = V_0 \uplus \dots \uplus V_\ell$ and $W = W_1 \uplus \dots \uplus W_\ell$ such that $N_G(w) \cap V_{i-1} = r$ for all $w \in W_i$ and all $i \in [\ell]$.
\end{corollary}

\begin{proof}
 Choose $R_0 \coloneqq \max(9,R_0')$ where $R_0'$ is the constant from Corollary \ref{cor:single-neighbor-expander} and suppose $r \geq R_0$, $\ell \geq 1$, and $m \geq 4r$.
 By Corollary \ref{cor:single-neighbor-expander}, there is a $(\frac{1}{4}r,\frac{1}{20r})$-single-neighbor expander $G' = (V',W',E')$ such that $|V'| = |W'| = m$ and $\deg_{G'}(w') = r$ for all $w' \in W'$.
 Suppose $V' = \{v_1',\dots,v_m'\}$ and $W' = \{w_1',\dots,w_m'\}$.
 
 We set $V_i \coloneqq \{v_{i,1},\dots,v_{i,m}\}$ for all $i \in [0,\ell]$ and $W_i \coloneqq \{w_{i,1},\dots,w_{i,m}\}$ for all $i \in [\ell]$.
 Also, we set
 \[E \coloneqq \{v_{i-1,j}w_{i,k} \mid i \in [\ell], v_j'w_k' \in E'\} \cup \{v_{i,j}w_{i,j} \mid i \in [\ell], j \in [m]\}.\]
 Clearly, $G = (V,W,E)$ is an $(\ell \times m)$-layered graph.
 
 Let $\alpha \coloneqq \frac{1}{4}r$ and $\gamma \coloneqq \frac{1}{20r}$.
 Also let $Y \subseteq W$ such that $|Y| \leq \gamma m$.
 We define $Y_i \coloneqq Y \cap W_i$ for all $i \in [\ell]$.
 Observe that $|Y_i| \leq \gamma m$ for all $i \in [\ell]$ and $Y_1,\dots,Y_\ell$ forms a partition of $Y$.
 Let $\mathcal{I} \coloneqq \{i \in [\ell] \mid Y_i \neq \emptyset\}$.
 Since $G'$ is an $(\alpha,\gamma)$-single-neighbor expander, we conclude that
 \[|N^{*}(Y_i) \cap V_{i-1}| \geq \alpha |Y_i|\]
 for all $i \in \mathcal{I}$. 
 Moreover, since $G[V_{i} \cup W_{i}]$ is $1$-regular (i.e., a matching) for all $i \in [\ell]$, we conclude that
 \[|N^{*}(Y) \cap V_{i-1}| \geq \alpha |Y_i| - |Y_{i-1}|\]
 for all $i \in \mathcal{I}$ (we set $Y_0 \coloneqq \emptyset$).
 So overall
 \[|N^{*}(Y)| \geq \sum_{i \in \mathcal{I}} \alpha |Y_i| - |Y_{i-1}| \geq \sum_{i \in \mathcal{I}} (\alpha - 1) |Y_i| = (\alpha - 1)|Y|\]
 as desired.
\end{proof}

\subsection{Constraint Sets from Layered Expanders}

Now, we turn to the construction of constraint sets from layered graphs.
For a bipartite graph $G = (V,W,E)$ we define the XOR-constraint set $\CC_G \coloneqq \{(N(w),0) \mid w \in W\}$ over the variable set $V$.
Slightly abusing notation, for $C \subseteq V$, we shall also write $C \in \CC_G$ if $(C,0) \in \CC_G$.

The basic idea for the construction of the XOR-constraint set $\CC$ is to take a layered graph $G = (V,W,E)$ with partitions $V = V_0 \uplus \dots \uplus V_\ell$ and $W = W_1 \uplus \dots \uplus W_\ell$, and set
\[\CC \coloneqq \CC_G \cup \big\{(\{x\},0) \bigmid x \in V_0\big\} \cup \big\{(\{x_\ell\},1)\big\}\]
for some arbitrary $x_\ell \in V_\ell$.
It is not difficult to see that this constraint set is unsatisfiable.
Indeed, every variable in layer $V_0$ needs to be set to $0$, and if all variables in layer $V_{i-1}$ are set to $0$, then the constraints obtained from the vertices in $W_i$ enforce that every variable in layer $V_i$ needs to be set to $0$ as well (using that $G[V_i \cup W_i]$ is a matching).
This inductive argument can be easily turned into a winning strategy for Falsifier that requires $O(\ell)$ many rounds (assuming the degree of all vertices in $W$ is bounded by some absolute constant $d \leq k$ where $k$ denotes the number of variables available in the game).

Now, the central claim is that, if we start with a single-neighbor layered expander, this strategy is essentially optimal.
Let us suppose for the moment that only constraints from $\CC_G$ are present and consider the $k$-closure $\cl_k(\CC_G)$.
What we need to avoid is that $\cl_k(\CC_G)$ contains some constraint that is ``non-local''.
For example, if $\cl_k(\CC_G)$ would contain a constraint $(\{x_1,x_2,x_3\},0)$ such that $x_1,x_2 \in V_0$ and $x_3 \in V_\ell$, then Falsifier could use such a (derived) constraint to immediately conclude that certain variables in the last layer need to be set to $0$ and potentially follow a different strategy to win the game faster.
The main point is that, by using single-neighbor layered expanders, we ensure that all ``relevant'' constraints in $\cl_k(\CC_G)$ are ``local'', i.e., they can only contain variables of $O(k)$ consecutive layers.
(Here, the reader may note that if $|N(w_1) \cup N(w_2)| \leq k$ then $(N(w_1) \oplus N(w_2),0)$ is always contained in the closure even if $w_1$ and $w_2$ are far apart. However, in such a case, $N(w_1) \cap N(w_2) = \emptyset$ and the derived constraint $(N(w_1) \cup N(w_2),0)$ is not ``relevant'' since, whenever it is violated by a partial assignment, one of the constraints associated with $w_1$ or $w_2$ is also violated.)
This way, even when adding all constraints from $\cl_k(\CC_G)$ to the initial set, the best that Falsifier can do is essentially to follow the above inductive strategy (with the exception that Falsifier may skip up to $O(k)$ layers in one step which, however, does not cause any problems for our arguments).

For technical reasons, the formal arguments slightly deviate from the intuitive ideas described above.
To start, instead of working with the $k$-closure $\cl_k(\CC_G)$, it turns out to be more convenient to work with the following set. 

Let $\alpha > 1$ and $k \geq 1$.
We define the set
\[\cl_{k,\alpha}^{*}(\CC_G) \coloneqq \Bigg\{\Bigg(\bigoplus_{D \in \CD}D,0\Bigg) \Biggmid \CD \subseteq \CC_G,|\CD| \leq \frac{k}{\alpha}, \Big|\bigoplus_{D \in \CD}D\Big| \leq k\Bigg\}.\]
We remark that, for $\CD = \emptyset$, the constraint $(\emptyset,0)$ is added to $\cl_{k,\alpha}^{*}(\CC_G)$.
Observe that $\CC_G \subseteq \cl_{k,\alpha}^{*}(\CC_G)$ if $k \geq \alpha$ (which is always the case in our constructions).
So the next lemma implies that $\cl_k(\CC_G) \subseteq \cl_{k,\alpha}^{*}(\CC_G)$ if $G$ is a suitable single-neighbor layered expander.

\begin{lemma}
 \label{la:single-neighbor-layered-small-sum-closed}
 Suppose $\alpha > 1$ and $0 < \gamma < 1$.
 Let $G = (V,W,E)$ be an $(\alpha,\gamma)$-single-neighbor $(\ell \times m)$-layered expander such that $\deg(w) \leq d$ for all $w \in W$ and suppose $d \leq k \leq \frac{1}{2}\gamma m$.
 Then
 \[\attr_k\left(\cl_{k,\alpha}^{*}(\CC_G)\right) = \cl_{k,\alpha}^{*}(\CC_G).\]
\end{lemma}

\begin{proof}
 Let $\CC^{*} \coloneqq \cl_{k,\alpha}^{*}(\CC_G)$.
 Suppose $C \in \attr_k(\CC^{*})$, that is, there are $C_1,C_2 \in \CC^{*}$ such that $C \coloneqq C_1 \oplus C_2$ and $|C| \leq k$.
 By definition, there are integers $s,t \leq \frac{k}{\alpha}$ and $D_1,\dots,D_s,D_{s+1},\dots,D_{s+t} \in \CC_G$ such that $C_1 = D_1 \oplus \dots \oplus D_s$ and $C_2 = D_{s+1} \oplus \dots \oplus D_{s+t}$.
 Moreover, $D_1,\dots,D_s$ are pairwise distinct as well as $D_{s+1},\dots,D_{s+t}$ are pairwise distinct.
 We have $C = D_1 \oplus \dots \oplus D_{s+t}$.
 Let
 \[\CD \coloneqq \{D_1,\dots,D_s\} \oplus \{D_{s+1},\dots,D_{s+t}\}\] and let $Y \coloneqq \{w \in W \mid N(w) \in \CD\}$.
 Clearly, $C = \bigoplus_{D \in \CD} D$.
 Suppose towards a contradiction that $|\CD| > \frac{k}{\alpha}$.
 Then $|Y| > \frac{k}{\alpha}$ and moreover, $|Y| \leq s+t \leq 2\frac{k}{\alpha} \leq 2k \leq \gamma m$ and thus, $|N^{*}(Y)| \geq \alpha|Y| > k$. But on the other hand $N^{*}(Y) \subseteq C$ which implies that $|N^{*}(Y)| \leq k$.
 This is a contradiction.
 So $|\CD| \leq \frac{k}{\alpha}$ which implies that $C \in \CC^{*}$ as desired.
\end{proof}

\begin{lemma}
 \label{la:single-neighbor-layered-no-singletons}
 Suppose $\alpha > 1$ and $0 < \gamma < 1$.
 Let $G = (V,W,E)$ be an $(\alpha,\gamma)$-single-neighbor $(\ell \times m)$-layered expander such that $\deg(w) \leq d$ for all $w \in W$ and suppose $d \leq k \leq \frac{1}{2}\gamma m$.
 Then $|C| \geq 2$ for all $C \in \cl_{k,\alpha}^{*}(\CC_G)$ such that $C \neq \emptyset$.
\end{lemma}

\begin{proof}
 Let $C \in \cl_{k,\alpha}^{*}(\CC_G)$ such that $C \neq \emptyset$ and let $C_1,\dots,C_s \in \CC_G$ such that $C = C_1 \oplus \dots \oplus C_s$ for some $s \leq \frac{k}{\alpha} \leq k$.
 Furthermore, let $Y \coloneqq \{w \in W \mid \exists i \in [s] \colon N(w) = C_i\}$.
 Observe that $1 \leq |Y| \leq k \leq \gamma m$.
 Then $N^{*}(Y) \subseteq C$ and thus, $|C| \geq |N^{*}(Y)| \geq \alpha |Y| > 1$.
\end{proof}

Next, we prove that Falsifier wins the pebble game if we set all variables in layer $V_0$ to $0$, and a single variable in the last layer $V_\ell$ to $1$.
For technical reasons, we do not add $(\{x_\ell\},1)$ to the constraint set, but rather consider an initial assignment that assigns value $1$ to variable $x_\ell$.

\begin{lemma}
 \label{la:falsifier-win}
 Let $G = (V,W,E)$ be an $(\ell \times m)$-layered graph with partitions $V = V_0 \uplus \dots \uplus V_\ell$ and $W = W_1 \uplus \dots \uplus W_\ell$ such that $\deg(w) \leq k$ for all $w \in W$.
 Let $x_\ell \in V_\ell$ and suppose $\beta_\ell\colon\{x_\ell\} \rightarrow \{0,1\}$ is the partial assignment defined via $\beta_\ell(x_\ell) = 1$.
 Then Falsifier wins $\CG_k(W,\CC,\beta_\ell)$ where
 \[\CC \coloneqq \CC_G \cup \big\{(\{x\},0) \bigmid x \in V_0\big\}.\]
\end{lemma}

\begin{proof}
 We prove by induction on $i = 0,\dots,\ell$ that Falsifier wins $\CG_k(V,\CC,\beta_i)$ where $\beta_i$ is any partial assignment for which $\beta_i(x_i) = 1$ for some $x_i \in V_i$.
 
 The base case $i = 0$ is trivial since $(\{x_0\},0) \in \CC$ for every $x_0 \in V_0$.
 For the inductive step, suppose $i \in [\ell]$ and consider some partial assignment $\beta_i$ for which there is some $x_i \in V_i$ such that $\beta_i(x_i) = 1$.
 Since $G = (V,W,E)$ is an $(\ell \times m)$-layered graph, there is a unique vertex $w_i \in W_i$ such that $w_ix_i \in E$.
 Moreover, $N_G(w_i) \subseteq V_{i-1} \cup V_i$.
 If $N_G(w_i) = \{x_i\}$, then $(\{x_i\},0) \in \CC_G$ and Falsifier wins immediately.
 So suppose that $N_G(w_i) \cap V_{i-1} \neq \emptyset$.
 Since $\deg(w_i) \leq k$, Falsifier can move to a partial assignment $\beta_{i-1}\colon X_i \rightarrow \{0,1\}$ where $X_i = N_G(w_i)$ and $\beta_{i-1}(x_i) = 1$.
 If $\beta_{i-1}$ violates the XOR-constraint $(X_i,0)$, then Falsifier wins immediately.
 Otherwise, $\sum_{y \in X_i} \beta_{i-1}(y) = 0$.
 Together with the fact that $\beta_{i-1}(x_i) = 1$, this implies that there is some $x_{i-1} \in X_i \cap V_{i-1}$ such that $\beta_{i-1}(x_{i-1}) = 1$.
 So Falsifier wins by the induction hypothesis.
\end{proof}

The next lemma forms the key technical lemma stating that Falsifier requires a large number of rounds to win if the constraint set is obtained from a single-neighbor layered expander.

\begin{lemma}
 \label{la:verifier-win-few-rounds}
 Suppose $\alpha > 1$ and $0 < \gamma < 1$.
 Let $G = (V,W,E)$ be an $(\alpha,\gamma)$-single-neighbor $(\ell \times m)$-layered expander with partitions $V = V_0 \uplus \dots \uplus V_\ell$ and $W = W_1 \uplus \dots \uplus W_\ell$ such that $\deg(w) \leq d$ for all $w \in W$ and suppose $d \leq k \leq \frac{1}{2}\gamma m$.
 
 Let $x_\ell \in V_\ell$ and suppose $\beta_\ell\colon\{x_\ell\} \rightarrow \{0,1\}$ is the partial assignment defined via $\beta_\ell(x_\ell) = 1$.
 Then Verifier wins $\CG_k^{r-1}(V,\CC,\beta_\ell)$ where
 \[\CC \coloneqq \CC_G \cup \big\{(\{x\},0) \bigmid x \in V_0\big\}\]
 and $r \coloneqq \lfloor\ell/2k\rfloor$.
\end{lemma}

\begin{proof}
 Let
 \[\CC_G^{*} \coloneqq \cl_{k,\alpha}^{*}(\CC_G)\]
 and define $\CC^{*} \coloneqq \CC_G^{*} \cup \{(\{x\},0) \mid x \in V_0\}$.
 We show that Verifier wins $\CG_k^{r-1}(V,\CC^{*},\beta_\ell)$ which clearly implies the claim since $\CC \subseteq \CC^{*}$ (using that $\alpha \leq d \leq k$).
 By Lemma \ref{la:dupl-closure-to-consistency}, it suffices to show that $\beta_\ell$ violates no XOR-constraint from the set $\cl_k^{(r-1)}(\CC^{*})$, or equivalently $(\{x_\ell\},0) \notin \cl_k^{(r-1)}(\CC^{*})$ (note that all constraints in $\cl_k^{(r-1)}(\CC^{*})$ are of the form $(C,0)$).
 
 We define
 \[\CV_i \coloneqq \bigcup_{j = 0}^{2ik} V_j\]
 for all $i \in \{0,\dots,\lfloor\ell/2k\rfloor\}$.
 Finally, we define
 \[\CC_i^{*} \coloneqq \Big\{C \subseteq V \Bigmid |C| \leq k, C = D \oplus U \text{ for some } D \in \CC_G^{*}, U \subseteq \CV_i\Big\}\]
 for all $i \in \{0,\dots,\lfloor\ell/2k\rfloor\}$.
 
 \begin{claim}
  \label{claim:attr-update}
  $\attr_k(\CC_i^{*}) \subseteq \CC_{i+1}^{*}$ for all $i \in \{0,\dots,\lfloor\ell/2k\rfloor-1\}$.
 \end{claim}
 \begin{claimproof}
  Let $C_1,C_2 \in \CC_i^{*}$ such that $|C_1 \oplus C_2| \leq k$.
  Let $C \coloneqq C_1 \oplus C_2$.
  For $j \in \{1,2\}$ pick $D_j \in \CC_G^{*}$ and $U_j \subseteq \CV_i$ such that $C_j = D_j \oplus U_j$.
  Let $U' \coloneqq U_1 \oplus U_2$.
  Clearly, $U' \subseteq \CV_i$ and $C = D_1 \oplus D_2 \oplus U'$.
  
  Let $Y_j \subseteq W$, $j \in \{1,2\}$, be a set of vertices of size $|Y_j| \leq \frac{k}{\alpha} < k$ such that $D_j = \bigoplus_{w \in Y_j} N(w)$ (recall that such a set $Y_j$ exists by the definition of $\cl_{k,\alpha}^{*}(\CC_G)$; for $D_j = \emptyset$ we set $Y_j \coloneqq \emptyset$).
  Then there is some $\lambda \in \{2ik+1,\dots,2(i+1)k\}$ such that $W_\lambda \cap (Y_1 \cup Y_2) = \emptyset$.
  We define
  \[Y_j^{< \lambda} \coloneqq Y_j \cap W_{< \lambda}\]
  where $W_{< \lambda} \coloneqq \bigcup_{\mu < \lambda} W_\mu$ and
  \[Y_j^{> \lambda} \coloneqq Y_j \cap W_{> \lambda}\]
  where $W_{> \lambda} \coloneqq \bigcup_{\mu > \lambda} W_\mu$.
  Moreover, let
  \[C_j^{> \lambda} \coloneqq \bigoplus_{w \in Y_j^{> \lambda}}N(w)\]
  for both $j \in \{1,2\}$.
  We have
  \[C_j^{> \lambda} \subseteq C_j\]
  because $C_j^{> \lambda} \subseteq D_j$ (since $W_\lambda \cap Y_j = \emptyset$) and $C_j^{> \lambda} \cap U_j = \emptyset$ (since $\lambda > 2ki$).
  Also let
  \[C_{> \lambda} \coloneqq C_1^{> \lambda} \oplus C_2^{> \lambda} \subseteq C.\]
  Hence, $|C_j^{> \lambda}| \leq k$ and $|C_{> \lambda}| \leq k$.
  So $C_j^{> \lambda} \in \CC_G^{*}$ for both $j \in \{1,2\}$.
  It follows that $C_{> \lambda} \in \CC_G^{*}$ by Lemma \ref{la:single-neighbor-layered-small-sum-closed}.
  
  Now, $C = C_{> \lambda} \oplus U$ for some $U \subseteq V_0 \cup \dots \cup V_{\lambda-1} \subseteq \CV_{i+1}$.
  It follows that $C \in \CC_{i+1}^{*}$.
 \end{claimproof}
 
 Since $\CC^{*} \subseteq \CC_0^{*}$ (this holds since $(\emptyset,0) \in \CC_G^{*}$) it follows by induction that
 \begin{equation}
  \cl_k^{(i)}(\CC^{*}) \subseteq \CC_i^{*}
 \end{equation}
 for all $i \in \{0,\dots,\lfloor\ell/2k\rfloor\}$ using Claim \ref{claim:attr-update}.
 So it only remains the prove the following claim.
 
 \begin{claim}
  $\{x_\ell\} \notin \CC_{r-1}^{*}$.
 \end{claim}
 \begin{claimproof}
  Let $C \in \CC_{r-1}^{*}$ such that $C \cap V_\ell \neq \emptyset$.
  Also pick $D \in \CC_G^{*}$ and $U \subseteq \CV_{r-1}$ such that $C = D \oplus U$ (which exist by the definition of $\CC_{r-1}^{*}$).
  We have that
  \[U \subseteq \CV_{r-1} = \bigcup_{i = 0}^{2k(r-1)} V_i \subseteq \bigcup_{i = 0}^{\ell - 2k} V_i.\]
  Let $Y \subseteq W$ such that $|Y| \leq \frac{k}{\alpha} < k$ and $D = \bigoplus_{w \in Y} N(w)$.
  Let $\lambda \in [\ell]$ be the maximal number such that $Y \cap W_\lambda = \emptyset$.
  Note that $\lambda > \ell - k$ since $|Y| < k$.
  Now let $D' \coloneqq \bigoplus_{w \in Y \cap W_{> \lambda}} N(w)$ where $W_{> \lambda} \coloneqq \bigcup_{\mu > \lambda} W_\mu$.
  Then $D' = C \cap (V_\lambda \cup \dots \cup V_\ell)$ and hence, $|D'| \leq k$.
  It follows that $D' \in \cl_{k,\alpha}^{*}(\CC_G)$.
  Also $|D'| \geq 1$ since $C \cap V_\ell \neq \emptyset$.
  So $|D'| \geq 2$ by Lemma \ref{la:single-neighbor-layered-no-singletons} and thus, $|C| \geq 2$.
 \end{claimproof}
\end{proof}

Finally, we require one more technical lemma that allows us to add the XOR-constraint $(\{x_\ell\},1)$ to the final constraint set. 

\begin{lemma}
 \label{la:strategy-assign-one-variable}
 Let $k \geq 2$ and $r \geq 1$.
 Let $V$ be a finite set and let $\CC$ be a set of XOR-constraints over $V$.
 Let $x_0 \in V$ and define $\beta_0\colon\{x_0\} \rightarrow \{0,1\}$ via $\beta_0(x_0) = 1$.
 If Verifier wins $\CG_k^{r}(V,\CC,\beta_0)$, then Verifier also wins $\CG_{k-1}^{r}(V,\CC \cup \{(\{x_0\},1)\},\emptyset)$.
\end{lemma}

\begin{proof}
 Consider a position $\beta\colon X \rightarrow \{0,1\}$ of the game $\CG_{k-1}^{r}(V,\CC \cup \{(\{x_0\},1)\},\emptyset)$.
 Throughout the game, by following a winning strategy for $\CG_k^{r}(V,\CC,\beta_0)$, Verifier can maintain the following properties after every round $\ell \in [0,r]$:
 \begin{enumerate}[label = (\roman*)]
  \item If $x_0 \in X$, then $\beta(x_0) = 1$, and
  \item Verifier wins the game $\CG_k^{r-\ell}(V,\CC,\beta')$ where $\beta' \colon X \cup \{x_0\} \rightarrow \{0,1\}$ is defined via $\beta'(x) \coloneqq \beta(x)$ for all $x \in X$ and $\beta(x_0) \coloneqq 1$.
 \end{enumerate}
 Observe that the condition is satisfied initially since Verifier wins $\CG_k^{r}(W,\CC,\beta_0)$.
 All positions reached this way clearly satisfy all XOR-constraints in $\CC \cup \{(\{x_0\},1)\}$ which implies that Verifier wins $\CG_{k-1}^{r}(W,\CC \cup \{(\{x_0\},1)\},\emptyset)$.
\end{proof}

With this, we are ready to prove Lemma \ref{la:many-rounds}.

\begin{proof}[Proof of Lemma \ref{la:many-rounds}]
 Let $R_0 \geq 9$ denote the constant from Corollary \ref{cor:single-neighbor-layered-expander} and define $\ell_\lo \coloneqq R_0 + 1$.
 Let $d \coloneqq R_0$, $\alpha \coloneqq \frac{1}{4}d - 1 > 1$ and $\gamma \coloneqq \frac{1}{20d}$.
 Let $\ell_\hi \geq \ell_\lo$ and $r \geq 1$ be given.
 We define $k \coloneqq \ell_{\hi} + 1$.
 Also, let $m \coloneqq 2 \cdot \frac{k}{\gamma} = 40dk \geq 4d$ and $\ell \coloneqq 2k(r+1)$.
 
 By Corollary \ref{cor:single-neighbor-layered-expander}, there is an $(\alpha,\gamma)$-single-neighbor $(\ell \times m)$-layered expander $G = (V,W,E)$ such that $\deg(w) = d+1$ for all $w \in W$.
 Let $V_0,\dots,V_\ell$ and $W_1,\dots,W_\ell$ denote the layers of $G$.
 Also fix some arbitrary element $x_\ell \in V_\ell$.
 We define
 \[\CC \coloneqq \CC_G \cup \{(\{x\},0) \mid x \in V_0\} \cup \{(\{x_\ell\},1)\}.\]
 Note that $\CC$ is a set of XOR-constraints over $V$ of arity at most $d+1 = \ell_\lo$.
 
 To complete the proof, we show that $\CC$ has the desired properties.
 First,
 \[|V| = (\ell + 1)m = (2k(r+1) + 1)40dk \leq 8kr \cdot 40dk = 320R_0 (\ell_\hi + 1)^{2}r \leq \delta \cdot \ell_\hi^{2} \cdot r\]
 for some suitable absolute constant $\delta$.
 Moreover, Falsifier wins the $\ell_\lo$-pebble game $\CG_{\ell_\lo}(V,\CC,\emptyset)$ by Lemma \ref{la:falsifier-win}.
 Finally, by Lemma \ref{la:verifier-win-few-rounds}, Verifier wins $\CG_k^{r}(V,\CC \setminus \{(\{x_\ell\},1)\},\beta_\ell)$ where $\beta_\ell\colon\{x_\ell\} \rightarrow \{0,1\}$ is the partial assignment defined via $\beta_\ell(x_\ell) = 1$.
 So Verifier wins the $r$-round $\ell_\hi$-pebble game $\CG_{\ell_\hi}^{r}(V,\CC,\emptyset)$ by Lemma \ref{la:strategy-assign-one-variable}.
\end{proof}

%% file: tradeoffs.tex
In this section, we investigate tradeoffs between the number of variables and the quantifier rank of formulas used to distinguish relational structures.
More concretely, suppose $\FA$ and $\FB$ are two structures of size $n$ that are distinguished by $k$-WL.
By Corollary \ref{cor:wl-logic-distinguish-structures}, there is a formula $\varphi \in \LCk{k+1}$ such that $\FA \models \varphi$ and $\FB \not\models \varphi$.
Using Theorem \ref{thm:wl-round-upper-bound}, we may assume that $\varphi$ has quantifier rank at most $O(kn^{k-1}\log n)$.
In this section, we show that there are sentences $\psi$ that distinguish between $\FA$ and $\FB$ with smaller quantifier rank if we are allowed to increase the number of variables by some function in $k$.
In other words, we can show improved bounds on the number of WL-iterations required to distinguish between $\FA$ and $\FB$ (compared to Theorem \ref{thm:wl-round-upper-bound}) by increasing the dimension of the WL-algorithm.

\begin{theorem}[Theorem \ref{thm:trading-upper-bound-intro} restated]
 \label{thm:trading-upper-bound}
 Let $k \geq 2$.
 Let $\FA$ and $\FB$ be two relational structures of arity at most $k$ such that $n \coloneqq |V(\FA)| = |V(\FB)|$.
 Also suppose there is a sentence $\varphi \in \LCk{k+1}$ such that $\FA \models \varphi$ and $\FB \not\models \varphi$.
 Let $d \coloneqq \lceil\frac{3(k+1)}{2}\rceil$.
 Then there is a sentence $\psi \in \LCkq{d}{q}$ of quantifier rank $q = O(k^2 \cdot n^{\lfloor k/2\rfloor + 1} \log n)$ such that $\FA \models \psi$ and $\FB \not\models \psi$.
\end{theorem}

Toward the proof of this theorem, let us fix some $k \geq 2$ and suppose that $k$ is odd, i.e., $k = 2\ell - 1$ for some integer $\ell \geq 2$ (this is the crucial case).
Let $\FA$ be a relational structure of arity at most $k$.
We translate $\FA$ into a binary structure (i.e., a structure of arity at most two) $\Binary(\FA)$ defined as follows.
The universe of $\Binary(\FA)$ is set to
\[V(\Binary(\FA)) \coloneqq (V(\FA))^\ell.\]
For every atomic type $\typ \in \{\atp_{\FA}(\vec v) \mid \vec v \in (V(\FA))^{2\ell}\}$ (on $2\ell$ vertices) we introduce a binary relation symbol $R_{\typ}$ and set
\[R_{\typ}^{\Binary(\FA)} \coloneqq \big\{\big((v_1,\dots,v_\ell),(v_{\ell+1},\dots,v_{2\ell})\big) \bigmid \atp_{\FA}(v_1,\dots,v_{2\ell}) = \typ\big\}.\]

Now, the key idea behind the proof of Theorem \ref{thm:trading-upper-bound} is to use $d$ variables to simulate the execution of $2$-WL on the binary structure $\Binary(\FA)$.
We can then obtain the upper bound on the quantifier rank by exploiting that $2$-WL stabilizes after at most $O(n \log n)$ rounds (see Theorem \ref{thm:wl-round-upper-bound}).

The next lemma translates a formula that distinguishes between $\Binary(\FA)$ and $\Binary(\FB)$ into a formula distinguishing $\FA$ and $\FB$.

\begin{lemma}
 \label{la:translate-formula-from-binary-structure}
 Let $\FA$ and $\FB$ be two relational structures of arity at most $k$.
 Suppose there is a sentence $\varphi \in \LCkq{d}{q}$ such that $\Binary(\FA) \models \varphi$ and $\Binary(\FB) \not\models \varphi$.
 Then there is a sentence $\widetilde{\varphi} \in \LCkq{d\cdot\ell}{q\cdot\ell}$ such that $\FA \models \widetilde{\varphi}$ and $\FB \not\models \widetilde{\varphi}$.
\end{lemma}

The proof of the lemma is a standard syntactic translation (see, e.g., \cite[Chapter 1.5]{Otto17}) and we omit the details here.

\begin{lemma}
 \label{la:wl-eq-from-binary-structure}
 Let $\FA$ and $\FB$ be two relational structures of arity at most $k$ such that $\Binary(\FA) \simeq_{2} \Binary(\FB)$.
 Then $\FA \simeq_k \FB$.
\end{lemma}

\begin{proof}
 Consider an arbitrary structure $\FC$ and define $\chi_2 \coloneqq \WL{2}{\Binary(\FC)}$ to be the coloring computed by $2$-WL on the structure $\Binary(\FC)$.
 We define a coloring $\chi\colon (V(\FC))^k \rightarrow C$ by setting
 \[\chi(v_1,\dots,v_k) \coloneqq \chi_2((v_1,\dots,v_\ell),(v_{\ell+1},\dots,v_{k-1},v_k,v_k)).\]
 \begin{claim}
  Suppose $\atp_\FC(v_1,\dots,v_k) \neq \atp_\FC(v_1',\dots,v_k')$.
  Then $\chi(v_1,\dots,v_k) \neq \chi(v_1',\dots,v_k')$.
 \end{claim}
 \begin{claimproof}
  Let $\typ \coloneqq \atp_\FC(v_1,\dots,v_{k-1},v_k,v_k)$.
  Then $((v_1,\dots,v_\ell),(v_{\ell+1},\dots,v_k,v_k)) \in R_{\typ}^{\Binary(\FC)}$, but on the other hand $((v_1',\dots,v_\ell'),(v_{\ell+1}',\dots,v_k',v_k')) \notin R_{\typ}^{\Binary(\FC)}$.
  So
  \[\atp_{\Binary(\FC)}((v_1,\dots,v_\ell),(v_{\ell+1},\dots,v_k,v_k)) \neq \atp_{\Binary(\FC)}((v_1',\dots,v_\ell'),(v_{\ell+1}',\dots,v_k',v_k'))\]
  which implies that
  \[\chi_2((v_1,\dots,v_\ell),(v_{\ell+1},\dots,v_k,v_k)) \neq \chi_2((v_1',\dots,v_\ell'),(v_{\ell+1}',\dots,v_k',v_k')).\]
  This directly implies the claim.
 \end{claimproof}
 \begin{claim}
  $\chi$ is $k$-stable.
 \end{claim}
 \begin{claimproof}
  Let $\vec v,\vec v' \in (V(\FC))^k$ such that $\chi(\vec v) = \chi(\vec v')$.
  Suppose $\vec v = (v_1,\dots,v_k)$ and $\vec v' = (v_1',\dots,v_k')$.
  Let us write $\vec v_1 \coloneqq (v_1,\dots,v_\ell)$ for the ``first half'' of $\vec v$, and $\vec v_2 \coloneqq (v_{\ell+1},\dots,v_k)$ for the ``second half''.
  Note that $\vec v_2$ has only $\ell - 1$ entries since $k = 2\ell-1$.
  Similarly, we define $\vec v_1' \coloneqq (v_1',\dots,v_\ell')$ and $\vec v_2' \coloneqq (v_{\ell+1}',\dots,v_k')$.
  For $w \in V(\FC)$ we write $\vec v_2 \circ w$ for the tuple $(v_{\ell+1},\dots,v_k,w)$ obtained from $\vec v_2$ by appending $w$.
  The tuple $\vec v_2' \circ w$ is defined analogously.
  
  Since $\chi_2$ is $2$-stable and $\chi_2(\vec v_1,\vec v_2 \circ v_k) = \chi_2(\vec v_1',\vec v_2' \circ v_k')$, we conclude that
  \[\Big\{\!\!\Big\{\big(\chi_2(\vec v_1,\vec w),\chi_2(\vec w,\vec v_2 \circ v_k)\big) \Bigmid \vec w \in (V(\FC))^\ell\Big\}\!\!\Big\} = \Big\{\!\!\Big\{\big(\chi_2(\vec v_1',\vec w),\chi_2(\vec w,\vec v_2' \circ v_k')\big) \Bigmid \vec w \in (V(\FC))^\ell\Big\}\!\!\Big\}.\]
  Using that $\chi_2$ refines the coloring by atomic types, it follows that
  \begin{align*}
           &\Big\{\!\!\Big\{\big(\chi_2(\vec v_1,\vec v_2 \circ w),\chi_2(\vec v_2 \circ w,\vec v_2 \circ v_k)\big) \Bigmid w \in V(\FC)\Big\}\!\!\Big\}\\
   = \quad &\Big\{\!\!\Big\{\big(\chi_2(\vec v_1',\vec v_2' \circ w),\chi_2(\vec v_2' \circ w,\vec v_2' \circ v_k')\big) \Bigmid w \in V(\FC)\Big\}\!\!\Big\}.
  \end{align*}
  In particular, we get that
  \[\Big\{\!\!\Big\{\chi_2(\vec v_1,\vec v_2 \circ w) \Bigmid w \in V(\FC)\Big\}\!\!\Big\} = \Big\{\!\!\Big\{\chi_2(\vec v_1',\vec v_2' \circ w) \Bigmid w \in V(\FC)\Big\}\!\!\Big\}.\]
  
  Now let $w,w' \in V(\FC)$ such that $\chi_2(\vec v_1,\vec v_2 \circ w) = \chi_2(\vec v_1',\vec v_2' \circ w')$.
  Then
  \[\chi(\vec v[w/i]) = \chi(\vec v'[w'/i])\]
  for all $i \in [k]$ using again that $\chi_2$ is $2$-stable and refines the coloring by atomic types.
  It follows that
  \[\Big\{\!\!\Big\{\big(\chi(\vec v[w/1]),\dots,\chi(\vec v[w/k])\big) \Bigmid w \in V(\FC)\Big\}\!\!\Big\} = \Big\{\!\!\Big\{\big(\chi(\vec v'[w/1]),\dots,\chi(\vec v'[w/k])\big) \Bigmid w \in V(\FC)\Big\}\!\!\Big\}.\]
  Overall, this implies that $\chi$ is $k$-stable.
 \end{claimproof}
 
 Combining both claims, we obtain that $\chi \preceq \WL{k}{\FC}$.
 Now, we complete the proof by setting $\FC$ to the disjoint union of $\FA$ and $\FB$.
\end{proof}

\begin{proof}[Proof of Theorem \ref{thm:trading-upper-bound}]
 First suppose that $k$ odd, i.e., $k = 2\ell-1$ for some integer $\ell \geq 2$.
 Since there is a sentence $\varphi \in \LCk{k+1}$ such that $\FA \models \varphi$ and $\FB \not\models \varphi$, we conclude that $\FA \not\simeq_k \FB$ using Corollary \ref{cor:wl-logic-distinguish-structures}.
 So $\Binary(\FA) \not\simeq_{2} \Binary(\FB)$ by Lemma \ref{la:wl-eq-from-binary-structure}.
 By Theorem \ref{thm:wl-round-upper-bound}, the $2$-WL algorithm distinguishes between $\Binary(\FA)$ and $\Binary(\FB)$ after at most $r = O(|V(\Binary(\FA))| \log |V(\Binary(\FA))|) = O(\ell \cdot n^\ell \cdot \log n)$ many refinement rounds.
 Using Corollary \ref{cor:wl-logic-distinguish-structures} again, this means there is a sentence $\varphi' \in \LCkq{3}{r}$ such that $\Binary(\FA) \models \varphi'$ and $\Binary(\FB) \not\models \varphi'$.
 So there is a sentence $\psi \in \LCkq{3\cdot\ell}{r\cdot\ell}$ such that $\FA \models \psi$ and $\FB \not\models \psi$ using Lemma \ref{la:translate-formula-from-binary-structure}.
 Note that $3\ell = 3 \cdot \frac{k+1}{2} = d$ and $r \cdot \ell = O(\ell^{2} \cdot n^\ell \cdot \log n) =  O(k^2 \cdot n^{(k+1)/2} \log n)$.
 
 For $k$ being even, the statement the of theorem follows by applying the first case to $k' = k+1$.
\end{proof}

%% file: conclusion.tex
We obtained new upper and lower bounds for the iteration number of the WL algorithm.
First, we showed that $k$-WL always stabilizes after at most $O(kn^{k-1}\log n)$ rounds for all $k \geq 2$, which is the first non-trivial upper bound on the iteration number for $k \geq 3$.
We complemented this result by a lower bound of $n^{\Omega(k)}$ which improves over the previously known lower bound of $n^{\Omega(k/\log k)}$~\cite{BerkholzN16}.
Finally, we also investigated tradeoffs between the dimension and the iteration number of WL.
Using known characterizations of WL, our results also imply upper and lower bounds on the quantifier rank of formulas in $\LCk{k}$ required to distinguish between two structures.

Still, several questions remain open.
The first question concerns the iteration number of $k$-WL on graphs.
The structures on which our lower bounds hold are $n$-element structures of arity $\Theta(k)$ and size $n^{\Theta(k)}$, and the increase in arity is inherent in the hardness condensation from~\cite{BerkholzN16}.
The best known lower bound on the iteration number of $k$-WL on graphs is $\Omega(n)$ due to F{\"{u}}rer \cite{Furer01}.
As an intermediate question, one can also ask for improved lower bounds in the size of the structure (i.e., the sum of the sizes of all relations), i.e., are there structures on which the iteration number of $k$-WL exceeds $\Omega(m)$ where $m$ denotes the size of the structure?

Our next question concerns the quantifier rank of formulas in $\LLk{k}$.
While our lower bounds extend to the logic $\LLk{k}$ (see Theorem \ref{thm:quantifier-rank-lower-bound}), this is not the case for the upper bounds that crucially rely on the availability of counting quantifiers.
A non-trivial upper bound of $O(n^{2}/ \log n)$ on the quantifier rank of formulas in $\LLk{3}$ has been obtained in \cite{KieferS19}.
Can we also obtain improved upper bounds on the quantifier rank of formulas in $\LLk{k}$ for $k \geq 4$?

Finally, we ask for further results on tradeoffs between the variable number and the quantifier rank.
Specifically, is there an integer $d \geq 3$ such that, for all structures $\FA$ and $\FB$ of size $n$ distinguished by $3$-WL, $d$-WL distinguishes between $\FA$ and $\FB$ in at most $\widetilde{O}(n)$ rounds (where $\widetilde{O}(\cdot)$ hides polylogarithmic factors)?
We remark that even $d = 3$ may be a valid choice, but any $d \geq 3$ is sufficient to obtain further tradeoffs in the spirit of Theorem \ref{thm:trading-upper-bound-intro}.

%% file: wl-rounds.bbl
\begin{thebibliography}{10}

\bibitem{Babai16}
L{\'{a}}szl{\'{o}} Babai.
\newblock Graph isomorphism in quasipolynomial time [extended abstract].
\newblock In Daniel Wichs and Yishay Mansour, editors, {\em Proceedings of the
  48th Annual {ACM} {SIGACT} Symposium on Theory of Computing, {STOC} 2016,
  Cambridge, MA, USA, June 18-21, 2016}, pages 684--697. {ACM}, 2016.
\newblock \href {https://doi.org/10.1145/2897518.2897542}
  {\path{doi:10.1145/2897518.2897542}}.

\bibitem{BabaiF20}
L\'{a}szl\'{o} Babai and P\'{e}ter Frankl.
\newblock {\em Linear algebra methods in combinatorics}.
\newblock University of Chicago, 2020.

\bibitem{BerkholzN16}
Christoph Berkholz and Jakob Nordstr{\"{o}}m.
\newblock Near-optimal lower bounds on quantifier depth and
  {W}eisfeiler-{L}eman refinement steps.
\newblock In Martin Grohe, Eric Koskinen, and Natarajan Shankar, editors, {\em
  Proceedings of the 31st Annual {ACM/IEEE} Symposium on Logic in Computer
  Science, {LICS} '16, New York, NY, USA, July 5-8, 2016}, pages 267--276.
  {ACM}, 2016.
\newblock \href {https://doi.org/10.1145/2933575.2934560}
  {\path{doi:10.1145/2933575.2934560}}.

\bibitem{CaiFI92}
Jin{-}yi Cai, Martin F{\"{u}}rer, and Neil Immerman.
\newblock An optimal lower bound on the number of variables for graph
  identification.
\newblock {\em Comb.}, 12(4):389--410, 1992.
\newblock \href {https://doi.org/10.1007/BF01305232}
  {\path{doi:10.1007/BF01305232}}.

\bibitem{Furer01}
Martin F{\"{u}}rer.
\newblock {W}eisfeiler-{L}ehman refinement requires at least a linear number of
  iterations.
\newblock In Fernando Orejas, Paul~G. Spirakis, and Jan van Leeuwen, editors,
  {\em Automata, Languages and Programming, 28th International Colloquium,
  {ICALP} 2001, Crete, Greece, July 8-12, 2001, Proceedings}, volume 2076 of
  {\em Lecture Notes in Computer Science}, pages 322--333. Springer, 2001.
\newblock \href {https://doi.org/10.1007/3-540-48224-5\_27}
  {\path{doi:10.1007/3-540-48224-5\_27}}.

\bibitem{Grohe08}
Martin Grohe.
\newblock The quest for a logic capturing {PTIME}.
\newblock In {\em Proceedings of the Twenty-Third Annual {IEEE} Symposium on
  Logic in Computer Science, {LICS} 2008, 24-27 June 2008, Pittsburgh, PA,
  {USA}}, pages 267--271. {IEEE} Computer Society, 2008.
\newblock \href {https://doi.org/10.1109/LICS.2008.11}
  {\path{doi:10.1109/LICS.2008.11}}.

\bibitem{Grohe17}
Martin Grohe.
\newblock {\em Descriptive Complexity, Canonisation, and Definable Graph
  Structure Theory}, volume~47 of {\em Lecture Notes in Logic}.
\newblock Cambridge University Press, 2017.
\newblock \href {https://doi.org/10.1017/9781139028868}
  {\path{doi:10.1017/9781139028868}}.

\bibitem{Grohe21}
Martin Grohe.
\newblock The logic of graph neural networks.
\newblock In {\em 36th Annual {ACM/IEEE} Symposium on Logic in Computer
  Science, {LICS} 2021, Rome, Italy, June 29 - July 2, 2021}, pages 1--17.
  {IEEE}, 2021.
\newblock \href {https://doi.org/10.1109/LICS52264.2021.9470677}
  {\path{doi:10.1109/LICS52264.2021.9470677}}.

\bibitem{GroheK21}
Martin Grohe and Sandra Kiefer.
\newblock Logarithmic weisfeiler-leman identifies all planar graphs.
\newblock In Nikhil Bansal, Emanuela Merelli, and James Worrell, editors, {\em
  48th International Colloquium on Automata, Languages, and Programming,
  {ICALP} 2021, July 12-16, 2021, Glasgow, Scotland (Virtual Conference)},
  volume 198 of {\em LIPIcs}, pages 134:1--134:20. Schloss Dagstuhl -
  Leibniz-Zentrum f{\"{u}}r Informatik, 2021.
\newblock \href {https://doi.org/10.4230/LIPIcs.ICALP.2021.134}
  {\path{doi:10.4230/LIPIcs.ICALP.2021.134}}.

\bibitem{GroheV06}
Martin Grohe and Oleg Verbitsky.
\newblock Testing graph isomorphism in parallel by playing a game.
\newblock In Michele Bugliesi, Bart Preneel, Vladimiro Sassone, and Ingo
  Wegener, editors, {\em Automata, Languages and Programming, 33rd
  International Colloquium, {ICALP} 2006, Venice, Italy, July 10-14, 2006,
  Proceedings, Part {I}}, volume 4051 of {\em Lecture Notes in Computer
  Science}, pages 3--14. Springer, 2006.
\newblock \href {https://doi.org/10.1007/11786986\_2}
  {\path{doi:10.1007/11786986\_2}}.

\bibitem{ImmermanL90}
Neil Immerman and Eric Lander.
\newblock Describing graphs: A first-order approach to graph canonization.
\newblock In Alan~L. Selman, editor, {\em Complexity Theory Retrospective: In
  Honor of Juris Hartmanis on the Occasion of His Sixtieth Birthday, July 5,
  1988}, pages 59--81. Springer New York, New York, NY, 1990.
\newblock \href {https://doi.org/10.1007/978-1-4612-4478-3_5}
  {\path{doi:10.1007/978-1-4612-4478-3_5}}.

\bibitem{KieferM20}
Sandra Kiefer and Brendan~D. McKay.
\newblock The iteration number of colour refinement.
\newblock In Artur Czumaj, Anuj Dawar, and Emanuela Merelli, editors, {\em 47th
  International Colloquium on Automata, Languages, and Programming, {ICALP}
  2020, July 8-11, 2020, Saarbr{\"{u}}cken, Germany (Virtual Conference)},
  volume 168 of {\em LIPIcs}, pages 73:1--73:19. Schloss Dagstuhl -
  Leibniz-Zentrum f{\"{u}}r Informatik, 2020.
\newblock \href {https://doi.org/10.4230/LIPIcs.ICALP.2020.73}
  {\path{doi:10.4230/LIPIcs.ICALP.2020.73}}.

\bibitem{KieferN22}
Sandra Kiefer and Daniel Neuen.
\newblock The power of the weisfeiler-leman algorithm to decompose graphs.
\newblock {\em {SIAM} J. Discret. Math.}, 36(1):252--298, 2022.
\newblock \href {https://doi.org/10.1137/20m1314987}
  {\path{doi:10.1137/20m1314987}}.

\bibitem{KieferS19}
Sandra Kiefer and Pascal Schweitzer.
\newblock Upper bounds on the quantifier depth for graph differentiation in
  first-order logic.
\newblock {\em Log. Methods Comput. Sci.}, 15(2), 2019.
\newblock \href {https://doi.org/10.23638/LMCS-15(2:19)2019}
  {\path{doi:10.23638/LMCS-15(2:19)2019}}.

\bibitem{LichterPS19}
Moritz Lichter, Ilia Ponomarenko, and Pascal Schweitzer.
\newblock Walk refinement, walk logic, and the iteration number of the
  {W}eisfeiler-{L}eman algorithm.
\newblock In {\em 34th Annual {ACM/IEEE} Symposium on Logic in Computer
  Science, {LICS} 2019, Vancouver, BC, Canada, June 24-27, 2019}, pages 1--13.
  {IEEE}, 2019.
\newblock \href {https://doi.org/10.1109/LICS.2019.8785694}
  {\path{doi:10.1109/LICS.2019.8785694}}.

\bibitem{MorrisLMRKGFB21}
Christopher Morris, Yaron Lipman, Haggai Maron, Bastian Rieck, Nils~M. Kriege,
  Martin Grohe, Matthias Fey, and Karsten~M. Borgwardt.
\newblock {W}eisfeiler and {L}eman go machine learning: The story so far.
\newblock {\em CoRR}, abs/2112.09992, 2021.
\newblock URL: \url{https://arxiv.org/abs/2112.09992}, \href
  {http://arxiv.org/abs/2112.09992} {\path{arXiv:2112.09992}}.

\bibitem{MorrisRFHLRG19}
Christopher Morris, Martin Ritzert, Matthias Fey, William~L. Hamilton, Jan~Eric
  Lenssen, Gaurav Rattan, and Martin Grohe.
\newblock {W}eisfeiler and {L}eman go neural: Higher-order graph neural
  networks.
\newblock In {\em The Thirty-Third {AAAI} Conference on Artificial
  Intelligence, {AAAI} 2019, The Thirty-First Innovative Applications of
  Artificial Intelligence Conference, {IAAI} 2019, The Ninth {AAAI} Symposium
  on Educational Advances in Artificial Intelligence, {EAAI} 2019, Honolulu,
  Hawaii, USA, January 27 - February 1, 2019}, pages 4602--4609. {AAAI} Press,
  2019.
\newblock \href {https://doi.org/10.1609/aaai.v33i01.33014602}
  {\path{doi:10.1609/aaai.v33i01.33014602}}.

\bibitem{MotwaniR95}
Rajeev Motwani and Prabhakar Raghavan.
\newblock {\em Randomized Algorithms}.
\newblock Cambridge University Press, 1995.
\newblock \href {https://doi.org/10.1017/cbo9780511814075}
  {\path{doi:10.1017/cbo9780511814075}}.

\bibitem{Neuen21}
Daniel Neuen.
\newblock Isomorphism testing parameterized by genus and beyond.
\newblock In Petra Mutzel, Rasmus Pagh, and Grzegorz Herman, editors, {\em 29th
  Annual European Symposium on Algorithms, {ESA} 2021, September 6-8, 2021,
  Lisbon, Portugal (Virtual Conference)}, volume 204 of {\em LIPIcs}, pages
  72:1--72:18. Schloss Dagstuhl - Leibniz-Zentrum f{\"{u}}r Informatik, 2021.
\newblock \href {https://doi.org/10.4230/LIPIcs.ESA.2021.72}
  {\path{doi:10.4230/LIPIcs.ESA.2021.72}}.

\bibitem{Neuen22}
Daniel Neuen.
\newblock Isomorphism testing for graphs excluding small topological subgraphs.
\newblock In Joseph~(Seffi) Naor and Niv Buchbinder, editors, {\em Proceedings
  of the 2022 {ACM-SIAM} Symposium on Discrete Algorithms, {SODA} 2022, Virtual
  Conference / Alexandria, VA, USA, January 9 - 12, 2022}, pages 1411--1434.
  {SIAM}, 2022.
\newblock \href {https://doi.org/10.1137/1.9781611977073.59}
  {\path{doi:10.1137/1.9781611977073.59}}.

\bibitem{Otto17}
Martin Otto.
\newblock {\em Bounded Variable Logics and Counting: {A} Study in Finite
  Models}, volume~9 of {\em Lecture Notes in Logic}.
\newblock Cambridge University Press, 2017.
\newblock \href {https://doi.org/10.1017/9781316716878}
  {\path{doi:10.1017/9781316716878}}.

\bibitem{ShervashidzeSLMB11}
Nino Shervashidze, Pascal Schweitzer, Erik~Jan van Leeuwen, Kurt Mehlhorn, and
  Karsten~M. Borgwardt.
\newblock {W}eisfeiler-{L}ehman graph kernels.
\newblock {\em J. Mach. Learn. Res.}, 12:2539--2561, 2011.
\newblock URL: \url{https://dl.acm.org/doi/10.5555/1953048.2078187}.

\bibitem{SunW15}
Xiaorui Sun and John Wilmes.
\newblock Faster canonical forms for primitive coherent configurations:
  Extended abstract.
\newblock In Rocco~A. Servedio and Ronitt Rubinfeld, editors, {\em Proceedings
  of the Forty-Seventh Annual {ACM} on Symposium on Theory of Computing, {STOC}
  2015, Portland, OR, USA, June 14-17, 2015}, pages 693--702. {ACM}, 2015.
\newblock \href {https://doi.org/10.1145/2746539.2746617}
  {\path{doi:10.1145/2746539.2746617}}.

\bibitem{Vadhan12}
Salil~P. Vadhan.
\newblock Pseudorandomness.
\newblock {\em Found. Trends Theor. Comput. Sci.}, 7(1-3):1--336, 2012.
\newblock \href {https://doi.org/10.1561/0400000010}
  {\path{doi:10.1561/0400000010}}.

\bibitem{Verbitsky07}
Oleg Verbitsky.
\newblock Planar graphs: Logical complexity and parallel isomorphism tests.
\newblock In Wolfgang Thomas and Pascal Weil, editors, {\em {STACS} 2007, 24th
  Annual Symposium on Theoretical Aspects of Computer Science, Aachen, Germany,
  February 22-24, 2007, Proceedings}, volume 4393 of {\em Lecture Notes in
  Computer Science}, pages 682--693. Springer, 2007.
\newblock \href {https://doi.org/10.1007/978-3-540-70918-3\_58}
  {\path{doi:10.1007/978-3-540-70918-3\_58}}.

\bibitem{WeisfeilerL68}
Boris Weisfeiler and Andrei Leman.
\newblock The reduction of a graph to canonical form and the algebra which
  appears therein.
\newblock {\em NTI, Series 2}, 1968.
\newblock English translation by Grigory Ryabov available at
  \url{https://www.iti.zcu.cz/wl2018/pdf/wl_paper_translation.pdf}.

\bibitem{XuHLJ19}
Keyulu Xu, Weihua Hu, Jure Leskovec, and Stefanie Jegelka.
\newblock How powerful are graph neural networks?
\newblock In {\em 7th International Conference on Learning Representations,
  {ICLR} 2019, New Orleans, LA, USA, May 6-9, 2019}. OpenReview.net, 2019.
\newblock URL: \url{https://openreview.net/forum?id=ryGs6iA5Km}.

\bibitem{Zimmermann14}
Alexander Zimmermann.
\newblock {\em Representation theory}, volume~19 of {\em Algebra and
  Applications}.
\newblock Springer, Cham, 2014.
\newblock A homological algebra point of view.
\newblock \href {https://doi.org/10.1007/978-3-319-07968-4}
  {\path{doi:10.1007/978-3-319-07968-4}}.

\end{thebibliography}
